\documentclass{llncs}
\usepackage{tabularx}
\usepackage{amsmath,amssymb}
\usepackage{subfig}
\usepackage{paralist}
\usepackage{hhline}
\usepackage{algorithm}
\usepackage{algpseudocode}
\usepackage{dsfont}
\usepackage{pifont}
\usepackage{graphicx}
\usepackage{color}
\usepackage{xcolor}
\usepackage{caption}
\usepackage{xspace}
\usepackage{tikz}
\usepackage{times}
\usepackage[scaled=.92]{helvet}
\usetikzlibrary{arrows,calc,positioning}
\usepgflibrary{shapes.multipart}

\usepackage{listings}
\usepackage{pgfplots}
\usepackage{wrapfig}
\usepackage{booktabs}
\usepackage{url}
\usepackage{pdflscape}

\renewcommand{\paragraph}[1]{\medskip\noindent 
{\bf #1.}\ }

\newcommand{\asgn}{\ensuremath{{\sf :=}}}
\newcommand{\yes}{{\tt Y}\xspace}
\newcommand{\no}{{\tt N}\xspace}

\newcommand{\states}{\ensuremath{\mathsf{states}}\xspace}
\newcommand{\guard}{\ensuremath{\mathsf{guard}}\xspace}

\newcommand{\stmt}{\ensuremath{\mathsf{stmt}}\xspace}

\newcommand{\stmts}{\ensuremath{\mathsf{stmts}}\xspace}

\newcommand{\hb}[2]{\ensuremath{\mathsf{hb}({#1},{#2})}}
\newcommand{\rdvar}[2]{\ensuremath{\mathsf{rd}({#1},{#2})}}
\newcommand{\wrvar}[2]{\ensuremath{\mathsf{wr}({#1},{#2})}}
\newcommand{\waw}[3]{\ensuremath{\mathsf{waw}_{{#1}}({#2},{#3})}}
\newcommand{\war}[3]{\ensuremath{\mathsf{war}_{{#1}}({#2},{#3})}}
\newcommand{\raw}[3]{\ensuremath{\mathsf{raw}_{{#1}}({#2},{#3})}}

\newcommand{\true}{\ensuremath{\mathsf{true}}\xspace}
\newcommand{\false}{\ensuremath{\mathsf{false}}\xspace}
\newcommand{\defn}{\ensuremath{\stackrel{\textup{\tiny def}}{=}}}

\def\var{\mathrm{Var}}
\def\In{\mathsf{in}}
\def\Out{\mathsf{out}}

\newcommand{\ConCrest}{{\sf ConCrest}\xspace}
\newcommand{\Concbugassist}{{\sf ConcBugAssist}\xspace}
\newcommand{\Bugassist}{{\sf BugAssist}\xspace}

\newcommand{\project}[2]{\ensuremath{#2\mathord{\downharpoonright}_{{#1}}}}

\newcommand{\enc}{\ensuremath{\mathsf{enc}}\xspace}

\newcommand{\fsenc}[1]{\ensuremath{\mathsf{csenc}({#1})}}
\newcommand{\hsenc}[1]{\ensuremath{\mathsf{hsenc}({#1})}}

\newcommand{\acquire}{\ensuremath{\mathsf{acquire}}\xspace}
\newcommand{\release}{\ensuremath{\mathsf{release}}\xspace}

\begin{document}

\pagestyle{plain}


\title{Error Invariants for Concurrent Traces}

\author{
Andreas Holzer\inst{1}
\thanks{Funded by the Erwin Schr\"odinger Fellowship J3696-N26 of the Austrian Science Fund (FWF).} \and 
Daniel Schwartz-Narbonne\inst{2}\thanks{Research was performed at NYU.} \and
Mitra Tabaei Befrouei\inst{3} 
\thanks{Supported by the Austrian National Research
Network S11403-N23 (RiSE), the LogiCS doctoral program 
W1255-N23 of the Austrian Science Fund (FWF) and
by the Vienna Science and Technology Fund (WWTF) through
grant VRG11-005.} 
\and\\
Georg Weissenbacher\inst{3}$^{\star\,\star\,\star}$\and
Thomas Wies\inst{4} 
\thanks{Funded in part by
the National Science Foundation under grant
CCF-1350574.}
}

\institute{
University of Toronto
\and
Amazon
\and
TU Wien
\and
New York University
}

\maketitle

\begin{abstract}

Error invariants are assertions that over-approximate the reachable
program states at a given position in an error trace while only
capturing states that will still lead to failure if execution of the
trace is continued from that position. Such assertions reflect the
effect of statements that are involved in the root cause of an error
and its propagation, enabling slicing of statements that do not
contribute to the error. Previous work on error invariants focused on
sequential programs. We generalize error invariants to concurrent
traces by augmenting them with additional information about hazards
such as write-after-write events, which are often involved in race
conditions and atomicity violations.
By providing the option to include varying levels of details in error
invariants---such as hazards and branching information---our approach
allows the programmer to systematically analyze individual aspects of
an error trace.  We have implemented a hazard-sensitive slicing tool
for concurrent traces based on error invariants and evaluated it on
benchmarks covering a broad range of real-world concurrency bugs.
Hazard-sensitive slicing significantly reduced the length of the
considered traces and still maintained the root causes of the
concurrency bugs.
\end{abstract}


\section{Introduction}
\label{sec:introduction}
\newcommand{\arrayvar}{a}

Debugging is notoriously time consuming. 
Once a program failure has been observed, the developer must identify
a cause-effect chain of events that led to it. This task is complicated 
by the fact that the underlying failing execution trace can contain 
a large number of events that do not contribute to the failure.

Error invariants \cite{ErmisSW12,Christ:flowSensitive,Murali:hybrid}
are (automatically generated) annotations of a given failing execution
trace that can support the developer in his endeavor to narrow down
the statements involved in the failure. Error invariants provide,
for each point in the trace, an over-approximation of the reachable 
states that will produce a failure if execution of the trace is continued from
that point (cf.~Definition~\ref{def:err_inv}). Consequently, two subsequent error invariants in an erroneous 
execution reflect the relevance of the interjacent statement to the observed failure.
Statements that leave the error invariant unchanged do not
contribute to the failure and can be safely ignored during the
failure analysis \cite{Murali:hybrid}.
 
Intuitively, failure analysis with error invariants can be understood as
a variant of dynamic slicing~\cite{Tip95asurvey} that takes the
semantics of the failure into account. Existing dynamic slicing
techniques are based on data- and control-flow dependencies and
remove statements which can not impact the failing state via 
any chain of dependencies. However, compared to error invariants
the precision of these syntax-based slicing techniques
is limited by the fact that the semantics of the erroneous trace
is not taken into account.

Error invariants have been successfully deployed for constructing 
semantics-aware slices in sequential software. The enabling techniques for the automated
generation of error invariants and slicing are \emph{unsatisfiable cores}
and \emph{interpolation}. An error trace translated into an unsatisfiable first-order logical
formula yields a proof of unsatisfiability from which interpolants can be extracted. 
These interpolants which correspond to assertions representing the error invariants can be 
used to construct a slice of the error trace that abstracts from the irrelevant statements and
explains the faulty behavior. This approach produces a slice
of the original trace annotated with assertions (the obtained
error invariants) showing the relevant values and variables to the
failure.

\begin{figure}[t]
	\centering
    	\small
    	\begin{minipage}[t]{0.45\textwidth}
    	\centering
    	\vspace{0pt}
    	\begin{tabularx}{\textwidth}{lX}
    	\multicolumn{2}{c}{\sffamily Code fragment-Deposit: $T_1$}\\
    	\toprule
    	& \scalebox{0.65}{\vdots}\\
    	& $\acquire~\ell$;\\
    	& \scalebox{0.65}{\vdots}\\
    	$L_1$: & $\mathsf{bal}:=\boldsymbol{\mathsf{balance}}$;\\
    	& $\release~\ell$;\\[5pt]
    	& {\sffamily if (bal+\arrayvar[i]$\leq$MAX})\\
      & \hspace{1em}{\sffamily bal $=$ bal+\arrayvar[i]};\\[5pt]
      & $\acquire~\ell$;\\
      $L_2$: & ${\boldsymbol{\mathsf{balance}}:=\mathsf{bal}}$;\\
      & \scalebox{0.65}{\vdots}\\
      & $\release~\ell$;\\
      & \scalebox{0.65}{\vdots}
      \end{tabularx}
      \end{minipage}%
      \hfill
      \begin{minipage}[t]{0.45\textwidth}
      \centering
      \vspace{0pt}
      \begin{tabularx}{\textwidth}{lX}
      \multicolumn{2}{c}{\sffamily Code fragment-Withdrawal: $T_2$}\\
    	\toprule
    	& \scalebox{0.65}{\vdots}\\
    	& $\acquire~\ell$;\\
    	& \scalebox{0.65}{\vdots}\\
    	$L'_1$: & $\mathsf{bal}:=\boldsymbol{\mathsf{balance}}$;\\
    	& $\release~\ell$;\\[5pt]
    	& {\sffamily if (bal-\arrayvar[j]$\geq$MIN})\\
      & \hspace{1em}{\sffamily bal $=$ bal-\arrayvar[j]};\\[5pt]
      & $\acquire~\ell$;\\
      $L'_2$: & ${\boldsymbol{\mathsf{balance}}:=\mathsf{bal}}$;\\
      & \scalebox{0.65}{\vdots}\\
      & $\release~\ell$;\\
      & \scalebox{0.65}{\vdots}
    	\end{tabularx}
    	\end{minipage}
    	\vspace{-1em}
  \caption{Non-atomic update of bank account balance\label{fig:bank_example}}
  \vspace{-2em}
\end{figure}

\paragraph{Error Invariants for Concurrent Traces} While
error invariants faithfully reflect sequential control- and data-flow,
concurrency aspects are ignored entirely. Consequently, a naive
application of error invariants to concurrent traces leads to
undesirable slices.

Consider, for example, the code fragments in Figure~\ref{fig:bank_example}.
At locations~$L_2$ and $L'_2$, respectively, threads~$T_1$ and~$T_2$ 
update the balance of a bank account which is stored in the shared variable 
{\sf balance}. 
The array {\sf \arrayvar} contains the sequence of 5 amounts to be transferred, 
partitioned into three deposits ($1\leq i\leq 3$) and two withdrawals
($4\leq j\leq 5$) executed by thread $T_1$ and $T_2$ in parallel,
respectively. Figure~\ref{fig:errinv_haz} shows the suffix of
a failing interleaved execution in which the third deposit is lost because
of an atomicity violation.
After three successful transactions (two deposits and one withdrawal)
thread $T_2$ stores the current {\sf balance} in a thread-local
variable {\sf bal}. 
At this point, $T_1$ interferes and updates the value of {\sf balance} by
performing the third deposit.
Thread~$T_2$, then, proceeds with the now stale value stored in {\sf bal}
and stores the result of the last withdrawal transaction in {\sf balance}.
Consequently, the execution results in a discrepancy of the expected and the 
actual balance on the account.

The problem is that the final value of {\sf balance} depends on the
sequence (or timing) of concurrently executed statements, i.e., the
program contains a \emph{data hazard}. As the statements are not
executed in the order expected by the programmer, the hazard results
in an erroneous state, which propagates to the end of the program where
it surfaces as a failure. In this setting, the fault the programmer
is looking for is the above-mentioned data hazard, in particular
the write-after-write dependency between $L_2$ and $L'_2$.

The gray assertions in Figure~\ref{fig:errinv_haz} represent
error invariants computed using the approach we propose in this paper.
The assertion after $L'_1$ states that the local variable {\sf bal}
reflects at most two deposits and one withdrawal. At this point,
the fault has not been triggered yet. The last conjunct in the error invariant
after the context switch indicates that the value of {\sf bal} is
unchanged. The error invariants produced by previous
techniques~\cite{ErmisSW12,Christ:flowSensitive,Murali:hybrid} 
track only the state information captured by this final conjunct.  
Therefore they would slice away all the statements of thread $T_1$
since the error invariants before and after the context switch would
be identical. Thus, the resulting slice
would not reflect the data hazard and not even the relevant interleaving.

To address this shortcoming, we lift interpolation-based slicing
techniques to a concurrency setting by taking into account control and
data dependencies between threads. The second assertion in $T_2$
(after the context switch) already reflects this adaptation: the
expression $\hb{L'_1}{L_2}\wedge\hb{L_2}{L'_2}$ indicates that the
statement at $L'_1$ happened before the statement at $L_2$, which in
turn happened before the one at $L'_2$. This specific order
is crucial to the failure. A slicing algorithm taking this information
into account cannot safely slice the statement at $L_2$ in thread
$T_1$ anymore. 
Note that, unlike previous techniques, 
error invariants in our approach not only reflect a set of states but also the
execution order of critical statements via the happens-before relation
(cf.~Section~\ref{subsec:data_dep}).

\begin{figure}[t]
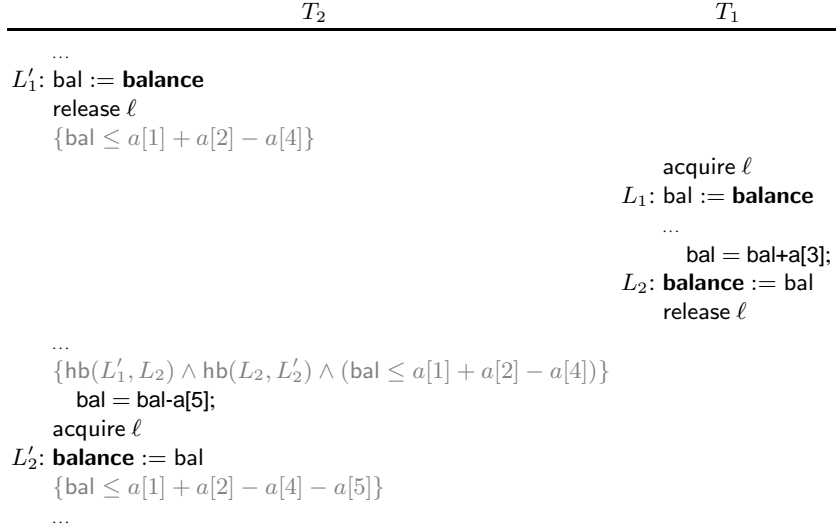
\centering
\begin{minipage}[t]{\textwidth}
\vspace{0pt}
\centering
{
\small
\begin{tabular}{ll@{\hspace{.45em}}ll}
\multicolumn{2}{c}{$T_2$} & \multicolumn{2}{c}{$T_1$}\\
\toprule
&\scalebox{0.65}{\dots}\\
$L'_1$: & $\mathsf{bal}:=\boldsymbol{\mathsf{balance}}$\\
& $\release~\ell$\\
& \textcolor{gray}{$\left\{\mathsf{bal}\leq
  \arrayvar[1]+\arrayvar[2]-\arrayvar[4]\right\}$}\\
&&& $\acquire~\ell$\\
&& $L_1$: & $\mathsf{bal}:=\boldsymbol{\mathsf{balance}}$\\
&&& \scalebox{0.65}{\dots}\\
&&&\hspace{1em}{\sffamily bal $=$ bal+\arrayvar[3]};\\
&& $L_2$: & $\boldsymbol{\mathsf{balance}}:=\mathsf{bal}$\\
&&& $\release~\ell$\\
&\scalebox{0.65}{\dots}\\
& \textcolor{gray}{$\left\{\hb{L'_1}{L_2}\wedge\hb{L_2}{L'_2}\wedge\left(\mathsf{bal}\leq
  \arrayvar[1]+\arrayvar[2]-\arrayvar[4]\right)\right\}$}\\
&\hspace{1em}{\sffamily bal $=$ bal-\arrayvar[5]};\\
&$\acquire~\ell$\\
$L'_2$: & $\boldsymbol{\mathsf{balance}}:=\mathsf{bal}$\\
& \textcolor{gray}{$\left\{\mathsf{bal}\leq
  \arrayvar[1]+\arrayvar[2]-\arrayvar[4]-\arrayvar[5]\right\}$}\\
&\scalebox{0.65}{\dots}
\end{tabular}
}
	\vspace{-.75em}
  \caption{Error trace with hazard-sensitive error invariants\label{fig:errinv_haz}}
	\vspace{-1.7em}
\end{minipage}
\end{figure}

Inter-thread data dependencies enable us to isolate (among other bugs)
race conditions and atomicity violations which constitute the predominant 
class of non-deadlock concurrency bugs~\cite{lu2008learning}.
Contrary to other concurrency debugging 
tools~\cite{DBLP:conf/sosp/EnglerA03,DBLP:journals/tocs/SavageBNSA97,FlanaganF10,FlanaganQ03,Park12,ParkVH10}
which target specific kinds of bugs,
we provide a general framework for concurrency bug explanation.
We applied an implementation of our approach to error traces 
generated from concurrent C programs using the directed testing tool 
\ConCrest~\cite{FarzanHRV13}.
We evaluate our approach on benchmarks that contain bugs found in real-world 
software such as Apache, GCC, and MySQL~\cite{Khoshnood2015}. 
On average, our slices yield a significant reduction of the number 
of variables and the length of the considered traces while maintaining
information that is crucial to understand the underlying concurrency bug.

\section{Preliminaries}
\label{sec:preliminaries}

\paragraph{Syntax of Concurrent Programs}
A concurrent program comprises multiple threads each represented by its control-flow graph
(CFG) \cite[\S 7]{muchnick}.

\begin{definition}[Control-Flow Graph] A CFG $\langle N, E\rangle$
  comprises nodes $N$ and edges $E$.
  Each node $n\in N$ corresponds to a single programming
  construct from a simple imperative language
  comprising assignments ${\sf x} \asgn e$
  and conditions $R$.

  Nodes representing
  conditional statements have two outgoing 
  edges labeled \yes and \no, respectively, corresponding to the
  positive and negative outcome of the condition. All other 
  nodes -- except the exit node, which has no successors --
  have out-degree one.
\end{definition}

\tikzstyle{block} = [rectangle, draw, minimum height=1.75em]
\tikzstyle{edge} = [->, >=stealth]
If a node $m$
is control dependent on a node $n$ and $n$ represents a condition, its
outcome can determine whether $m$ is reached:

\begin{definition}[Dominators and Control Dependency]
  \label{def:ctrl-dep}
  A node $m$ post-domi\-nates a node $n$ if all paths to the
  exit node starting at $n$ must go through $m$. Node $m$
  is control dependent on $n$ (where $n\neq m$) if
  $m$ does not post-dominate $n$ and there
  exists a path from $n$ to $m$ such that
  $m$ post-dominates all nodes (other than $n$) on that
  path.
\end{definition}

Based on Definition \ref{def:ctrl-dep}, we introduce our
notion of a scope: 
\begin{definition}[Scope]
  \label{def:scope}
  A node $m$ is in scope of the condition at node
  $n$ if $m$ is control dependent on $n$ or in scope of
  a condition that is control dependent on $n$.
\end{definition}
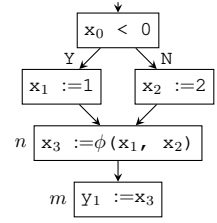
\begin{wrapfigure}[11]{r}{0.325\textwidth}
	\vspace{-2em}
	\centering
    \begin{tikzpicture}[scale=.75,every node/.style={scale=0.75}]
      \node [block] (cond) at (1,2) { \tt x$_0$ < 0 };
      \node [block] (then) at (0,1) { \tt x$_1$ \asgn 1 };
      \node [block] (else) at (2,1) { \tt x$_2$ \asgn 2 };
      \node [block,label=left:{$n$}] (phi) at (1,0) { 
        {\tt x$_3$ \asgn $\phi$(x$_1$, x$_2$)}};
      \node [block,label=left:{$m$}] (last) at (1,-1) { 
        {\tt y$_1$ \asgn x$_3$ }};
      
      \draw[edge] ($(cond)+(0,.5)$)--(cond);
      \draw[edge] (cond)--node[left=.1cm]{\yes}(then);
      \draw[edge] (cond)--node[right=.1cm]{\no}(else);
      \draw[edge] (then)--(phi);
      \draw[edge] (else)--(phi);
      \draw[edge] (phi)--(last);
    \end{tikzpicture}
    \caption{SSA form of: {\scriptsize \tt if~(x<0) then x\asgn 1 else x\asgn 2; y\asgn x}
      \label{fig:ssa}}
\end{wrapfigure}
A CFG is in Static Single Assignment (SSA) form \cite{cfrwz91}
if each variable is assigned exactly once.
The standard mechanism to translate CFGs into SSA form
is to subscript each definition of a variable with a unique version
number; consequently, each definition is uniquely identified by 
the corresponding SSA variable. 
Conflicting definitions at a control-flow merge point
$m$ in a CFG are resolved by 
introducing an arbiter node $n$ (with sole successor $m$) 
to which we divert the incoming edges of $m$. 
The arbiter node $n$ is annotated with a $\phi$-function
which switches between the definitions from different
incoming paths (see Figure \ref{fig:ssa}).
Algorithms to convert a program into SSA form are described
in \cite{cfrwz91} and \cite[\S 8.11]{muchnick}.

\begin{definition}[Program Path]
  Let $\langle N_t, E_t\rangle$ be a CFG representing a thread $t$.
  A path $P_t$ of thread $t$ is a sequence
  $n_1,\langle n_1,n_2\rangle,n_2$, \ldots, $\langle n_{k-1},n_k\rangle,n_k$
  of nodes $n_i\in N_t$ and edges $\langle n_i,n_{i+1}\rangle\in E_t$.
  A program path $P\defn n_1,\langle n_1,n_2\rangle,n_2$, \ldots, $\langle
  n_{k-1},n_k\rangle,n_k$
  corresponds to an interleaving of paths of threads
  (starting at their respective initial nodes)
  such that for each $i$ with $1\leq i<k$ either 
  $n_i, n_{i+1}\in N_t$ and 
  $\langle n_i, n_{i+1}\rangle\in E_t$ for some thread $t$,
  or $n_i$ and $n_{i+1}$ belong
  to different threads and $\langle n_i, n_{i+1}\rangle$ is an
  inter-thread edge representing a context switch.
\end{definition}

Given a (program) path $P$,
let $[n_i,n_j]$ denote the sub-path $n_i,$ $\langle
n_i,n_{i+1}\rangle$, $n_{i+1}$, \ldots, $n_{j-1}$, $\langle n_{j-1},n_j\rangle,n_j$
of $P$ including the nodes $n_i$ and $n_j$ and $(n_i,n_j)$
the sub-path $\langle n_i,n_{i+1}\rangle,n_{i+1}$, \ldots,
$n_{j-1},\langle n_{j-1},n_j\rangle$ excluding
the nodes $n_i$ and $n_j$. We use $\project{t}{P}$ to denote the
projection of a program path $P$ to thread $t$ in which only
nodes $n_i\in N_t$ and edges $\langle n_i,n_{i+1}\rangle\in E_t$ are
retained and any sub-path
$(n_i,n_j)$ with $n_i,n_j\in N_t$ and $n_l\notin N_t$ for $i<l<j$ is replaced
with the edge $\langle n_i,n_j\rangle$, i.e., $\project{t}{P}$ is a
path of the thread $t$. Consequently, for each program path $P$, 
$\project{t}{P}$ is either empty (if $P$ does not visit thread
$t$) or a path of thread $t$ starting at the initial node of $t$. 
Finally, $\project{N}{P}$ and $\project{E}{P}$
denote the projection of $P$ to the sequence of nodes $N$ and
edges $E$ in $P$, respectively. 

\paragraph{Semantics, Feasible Executions and Error Traces} 
The variables of a program are partitioned into
\emph{global} and \emph{thread-local} variables.
A state $s$ maps each variable to
a value, and $s(e)$ denotes the value of
expression $e$ in state $s$.

A program path $P$ corresponds to a sequence of statements. 
We require that each statement refers to at most one global variable, and hence statements execute atomically.  
\begin{definition}[Execution]
  An execution of a path $P$ corresponds to an execution of the
  statements of $P$ in order (starting in an initial state). We use
  $\stmt_P(n_i)$ to denote the statement represented by node $n_i$ in
  a path $P$. In particular, if node $n_i$ represents the condition
  $R$, let $t$ be such that $n_i\in N_t$ and let $\langle
  n_i,n_j\rangle$ be the first edge in $\project{t}{P}$ succeeding
  $n_i$.  Then $\stmt_P(n_i)$ is $R$ if $\langle n_i,n_j\rangle$ is
  labeled \yes, $\neg R$ if the edge is labeled \no. If
  $n_i$ is the last node of a thread $t$, then $\stmt_P(n_i)=\true$.

The execution of one statement in
the current program state $s$ is defined as follows:
\begin{compactitem}
\item If $\stmt_P(n_i)$ is the assignment ${\tt x}\asgn e$,
  the successor state of $s$ is updated such that
 ${\tt x}$ evaluates to $s(e)$ and all other variables
  are unchanged.
\item If $n_i$ is a conditional statement $R$,
  the execution proceeds iff $s(\stmt_P(n_i))$
  is true.
\end{compactitem}
\end{definition}

A path $P$ is \emph{feasible} if
there exists an initial state $s$ for which the
execution of $P$ is not blocked by a condition which is false. Given a path $P$,
we use $\stmts_P$ to denote the sequence of statements
represented by $P$. Abusing our notation,
we sometimes call $\stmts_P$ a path and will
use $P$ and $\stmts_P$ interchangeably. 

We use $\stmts_P[i]$ to denote
the $i^{\text{th}}$ statement $\stmt_P(n_i)$ of a path $P$,
and $\stmts_P[i,j]$ to denote the sub-path
$\stmts_P[i]{\sf;}$ $\ldots{\sf;}$ $\stmts_P[j]$
($[n_i,n_j]$, respectively).
We drop the subscript $P$ if it is clear from the context.

A state $s_j$ is reachable from a state $s_{i-1}$ via a sub-path
$\stmts_P[i,j]$ if an execution of $\stmts_P[i,j]$ starting
in $s_{i-1}$ does not block and results in state $s_j$. 

We assume that the correctness of a path is determined by an assertion
$\psi$ expected to hold after the execution of the path. Error traces which
result in the violation of $\psi$ are defined as:
\begin{definition}[Error Trace] A path $P$
  is an \emph{error trace} for the assertion $\psi$
  if $P$ is feasible and always results in a state $s$
  such that $s(\psi)$ is false.
  \label{def:error_trace}
\end{definition}
Intuitively, an error trace is an execution of a failing test case
that does not satisfy the specification $\psi$.
We assume (w.l.o.g.) that path $P$ in Definition
\ref{def:error_trace} reaches the 
end of the main thread, where $\psi$ is asserted. Consequently,
$\psi$ is not in scope of any condition.

\section{Error Explanation}
\label{sec:fault-localization-paths}
In this section, we first recall the interpolation-based slicing approach
presented in \cite{ErmisSW12,Christ:flowSensitive} for sequential software. 
We then explain how we extend it to concurrent executions. 

\subsection{Interpolation-based Slicing for Sequential Traces}
\label{sec:error_invariants}
Ermis et al.~\cite{ErmisSW12} and Christ
et al.~\cite{Christ:flowSensitive} use \emph{error invariants} to identify
statements that do not contribute to the assertion violation in sequential traces.

\begin{definition}[Error Invariant]
  Given an error trace $P$ of length $k$ for assertion $\psi$,
  an error invariant for position $i$ (with $i\leq k$) is a set of states
  $E$ such that
  \begin{compactitem}
  \item[(a)] $E$ contains (at least) all states
    reachable from an initial state via $\stmts_P[1,i]$, and
  \item[(b)] every feasible execution of $\stmts_P[i+1,k]$
    starting from a state in $E$ results in a state
    in which $\psi$ is false.
  \end{compactitem}
  An error invariant $E$ is
  \emph{recurring}\footnote{To avoid confusion with inductive
    interpolant sequences (Definition~\ref{def:ind_itp}),
    we replace the notion of \emph{inductive error
      invariants}~\cite{ErmisSW12,Christ:flowSensitive}
    with recurring error invariants.}
  for positions $i\leq j$ if $E$ is an
  error invariant for $i$ as well as for $j$.
  \label{def:err_inv}
\end{definition}

Intuitively, an error invariant $E$ represents an over-approx\-imation
of the states that are reachable via the path $\stmts[1,i]$ such that 
$\stmts[i+1,k]$ if executed from a state in $E$ still results
in failure. According to~\cite{ErmisSW12,Christ:flowSensitive},
statements between a recurring error invariant are
``not needed to reproduce the error.''

Error invariants can be derived using Craig
interpolation (defined below) and a symbolic encoding 
of a path $P$ \cite{ErmisSW12,Christ:flowSensitive}.
In the following, we derive a symbolic
encoding $\enc(P)$ similar to the one in \cite{ErmisSW12}
from a straight-line program in SSA form, which represents
the path $P$ to be encoded. This straight-line program
is obtained by traversing the CFG along $P$.
If a node is visited repeatedly (via a cycle in one
of the CFGs), a new version of the variable is introduced;
for straight-line programs
(which do not contain control-flow merge points)
it suffices to increase the version number of a variable
each time it is assigned and refer to the latest version of
each variable in conditions and right-hand sides of assignments.

Given a path $P$ in SSA form as described above, 
the formula $\enc(P)$ is a conjunction
$\bigwedge_{i=1}^k \enc_P(n_i)$ of the encodings of the
individual statements:
\begin{equation}
  \enc_P(n_i)\defn
  \left\{
  \begin{array}{ll}
    ({\sf x}_i=e) & \text{if}~\stmt_P(n_i)~\text{is}~{\sf
      x}_i\,\asgn\,e\\
    \stmt_P(n_i) & \text{if}~\stmt_P(n_i)~\text{is a condition}\\
  \end{array}
  \right.
  \label{eq:simple_enc}
\end{equation}

Variable assignments that satisfy formula $\enc(P)$ correspond
to executions; note that if all variables in $P$ are
initialized before being used, $\enc(P)$ has only
one unique satisfying assignment. In this context,
interpolants (Definition \ref{def:itp} below)
are a symbolic representation of sets of states.
Let $\var(A)$ be the set of (free) variables occurring in a formula $A$.
An interpolant $I$ is a predicate that encodes all
states $s$ for which $s(I)$ is true. We define
$\states(I)\defn\{s\,\vert\,s(I)=\true\}$.

\begin{definition}[Interpolant]
  Let $A$ and be $B$ be a pair of first-order formulas such that
  $A\wedge B$ is unsatisfiable. An \emph{interpolant} of
  $A$ and $B$ is a first-order formula $I$ such that
  $A\Rightarrow I$, $B\Rightarrow \neg I$, and
  $\var(I)\subseteq\var(A)\cap\var(B)$.
  \label{def:itp}
\end{definition}
Definition~\ref{def:itp} corresponds to the definition of
interpolants in~\cite{McMillan05} under the assumption that all
non-logical symbols in $A$ and $B$ are interpreted.

The following definition is a generalization of
interpolants:
\begin{definition}[Inductive Interpolant Sequence]
  Let $A_1$, \ldots, $A_n$ be a sequence of first-order formulas
  whose conjunction is unsatisfiable. Then $I_0,\ldots I_n$
  is an inductive interpolant sequence if
  \begin{compactitem}
  \item $I_0=\true$ and $I_n=\false$,
  \item for all $1\leq i \leq n$, $I_{i-1}\wedge A_i\Rightarrow I_i$,
    and
  \item for all $1\leq i<n$, $\var(I_i)\in
    (\var(A_1\wedge\ldots\wedge A_i)\cap\var(A_{i+1}\wedge\ldots\wedge A_n))$. 
  \end{compactitem}
  \label{def:ind_itp}
\end{definition}
\tikzset{
  states/.style={
    circle,
    draw,
    fill=white,
    minimum height=1cm,
    text centered},
}
Given a path
$P\defn n_1$, $\langle n_1,n_2\rangle,n_2$, \ldots,
$\langle n_{k-1},n_k\rangle$, $n_k$
in SSA form, a sequence interpolant $I_0,\ldots,I_{k+1}$
derived from the formulas $\enc_P(n_1)$, \ldots, $\enc_P(n_k)$, $\psi$
is inductive in the sense that 
$\states(I_{i})$ contains all states reachable from $\states(I_{i-1})$ via
$\stmt(n_i)$ (and potentially more) \cite{impact,Murali:hybrid}.
Moreover,
$I_k\wedge\psi$ is not satisfiable, i.e., all states represented
by $I_k$ violate assertion $\psi$.
If $I_i$ represents an error invariant for positions
$i$ and $j$ (i.e., $\states(I_i)$ is an error invariant for $j$
and $I_i$ implies $I_j$)
then $I_i$ is inductive with respect to the sub-path $\stmts_P[i+1,j]$.
Accordingly, slicing $[n_{i+1},n_j]$ away (i.e, replacing it
with an edge $\langle n_{i+1},n_j\rangle$) preserves the
assertion violation.

A trace obtained by removing statements between 
recurring error invariants from $P$ is sound in the sense of Definition
\ref{def:sound_slice} below:
\begin{definition}[Sound Slice]
  A slice of path $P$ of length $k$ is a
  path $Q$ of length $m$ with
  $\stmts_Q[1]=\stmts_P[i_1]$,
  $\stmts_Q[2]=\stmts_P[i_2]$, \ldots,
  $\stmts_Q[m]=\stmts_P[i_m]$ with
  $1\leq i_1<i_2<\ldots<i_m\leq k$. Given an error trace
  $P$ for $\psi$, a slice $Q$ of
  $P$ is \emph{sound} if 
  $Q$ is also an error trace for $\psi$.
  \label{def:sound_slice}
\end{definition}

\subsection{Interpolation-based Slicing for Concurrent Traces}
\label{subsec:concurrency}
In the following, we enhance and extend the interpolation-based
slicing technique discussed in Section \ref{sec:error_invariants}
to take control dependency as well as concurrency into account.

\paragraph{Control Dependencies}
\label{subsec:control_sensitive}
The following example shows that the encoding $\enc({\stmts})$
fails to capture control dependence (Definition \ref{def:ctrl-dep}).

\begin{example}
  Figure \ref{fig:cs_insensitive} shows the statements of a path
  $P$ (in SSA form) and a corresponding interpolant sequence on the right.
  The example is a sequential variation of the bank account example
  which fails if the required minimum balance {\sf MIN} is larger than zero.
  The resulting slice (indicated in bold) contains
  only the last assignment to {\sf bal} and the assertion $\psi$.
  It does not reflect the fact that the \yes-branch of
  the conditional statement has to be taken for the failure
  to occur.
  \label{ex:cf_insensitive}
\end{example}

\begin{figure}[t]\centering
\subfloat[Control-insensitive slice\label{fig:cs_insensitive}]{
\small
  \begin{tikzpicture}
    \node at (0,0) {
      \begin{minipage}{.8\columnwidth}
        \begin{tabbing}
          \quad\=\hspace{2.5cm}\=\kill
          $\mathsf{bal_1\asgn MIN;}$ \>\>\\ 
          $\mathsf{a_1\asgn -100;}$ \>\> \\ 
          $\mathsf{if (bal_1+a_1 \leq MIN)}$ \>\> \\
          \>\boldmath$\mathsf{bal_2\asgn 0;}$ \>\\
          \boldmath$\mathsf{assert(bal_2 \geq MIN);}$ 
        \end{tabbing}
    \end{minipage}};
    \draw[gray]
    (-1,2.5\baselineskip)--(1.75,2.5\baselineskip)node[right]{$\{\true\}$};
    \draw[gray]
    (-1,1.5\baselineskip)--(1.75,1.5\baselineskip)node[right]{$\{\true\}$};
    \draw[gray]
    (-1,.5\baselineskip)--(1.75,.5\baselineskip)node[right]{$\{\true\}$};
    \draw[gray]
    (-1,-.5\baselineskip)--(1.75,-.5\baselineskip)node[right]{$\{\true\}$};
    \draw[gray]
    (-1,-1.5\baselineskip)--(1.75,-1.5\baselineskip)node[right]{$\{ {\sf bal}_2=0 \}$};
    \draw[gray]
    (-1,-2.5\baselineskip)--(1.75,-2.5\baselineskip)node[right]{$\{\false\}$};
  \end{tikzpicture}
}%
\subfloat[Control-sensitive slice\label{fig:cs_sensitive}]{
\small
  \begin{tikzpicture}
    \node at (0,0) {
      \begin{minipage}{.8\columnwidth}
        \begin{tabbing}
          \quad\=\hspace{3cm}\=\kill
          \boldmath$\mathsf{bal_1\asgn MIN;}$ \>\>\\ 
          \boldmath$\mathsf{a_1\asgn -100;}$ \>\> \\ 
          $\mathsf{if (bal_1+a_1 \leq MIN)}$ \>\> \\
          \>\boldmath$\mathsf{bal_2\asgn 0;}$ \>\\
          \>\>\\
          \boldmath$\mathsf{bal_3\asgn \phi(bal_2);}$ \>\\
          \boldmath$\mathsf{assert(bal_3 \geq MIN);}$ 
        \end{tabbing}
    \end{minipage}};
    \draw[gray]
    (-1,3.5\baselineskip)--(1.75,3.5\baselineskip)node[right]{$\{\true\}$};
    \draw[gray]
    (-1,2.5\baselineskip)--(1.75,2.5\baselineskip)node[right]
         {$\{{\sf bal}_1\leq \mathsf{MIN} \}$};
    \draw[gray]
    (-1,1.5\baselineskip)--(1.75,1.5\baselineskip)node[right]
         {$\{{\sf bal}_1+{\sf a}_1\leq \mathsf{MIN}  \}$};
    \draw[gray]
    (-1,0.5\baselineskip)--(1.75,0.5\baselineskip)node[right]
         {$\{{\sf bal}_1+{\sf a}_1\leq \mathsf{MIN}  \}$};
    \draw[gray]
    (-1,-1\baselineskip)--(1.75,-1\baselineskip)node[right]
         {\scriptsize$\left\{
           \begin{array}{ll}({\sf bal}_1+{\sf a}_1\leq
             \mathsf{MIN})\cr\wedge({\sf bal}_2=0)\end{array}\right\}$};
    \draw[gray]
    (-1,-2.5\baselineskip)--(1.75,-2.5\baselineskip)node[right]{$\{{\sf bal}_3=0\}$};
    \draw[gray]
    (-1,-3.5\baselineskip)--(1.75,-3.5\baselineskip)node[right]{$\{\false\}$};
  \end{tikzpicture}
}
\caption{Slicing sequential trace with Error Invariants}
\vspace{-1.7em}
\end{figure}

We present a (modular) extension to the encoding 
defined in Section~\ref{sec:error_invariants} that enables the inclusion
of control dependencies. Unlike prior work \cite{Christ:flowSensitive},
which addresses this problem using
a custom-tailored control-sensitive encoding, our
technique is based on the SSA representation.
As in Section~\ref{sec:error_invariants},
the starting point of our approach is a straight-line
representation of the error trace $P$. Unlike before,
however, we include the $\phi$-nodes from the
SSA presentation of the program in $P$:
\begin{description}
\item[$\phi$-functions] at $n\in N_t$ for a variable {\sf x}, 
  take as parameters the subscripted variable versions
  representing definitions of {\sf x} in thread $t$ that reach $n$.
\end{description}
Consequently, when generating the straight-line presentation of $P$,
we include all $\phi$-nodes of the SSA presentation of the program
that are traversed by $P$. As we are encoding a single path $P$, 
however, $\phi$ takes only one parameter, since only one definition of each variable {\sf x} reaches $n$ in $P$.
Our extension $\fsenc{\stmts}$
of the encoding $\enc(\stmts)$ is based on 
assignments ${\sf x}_i\asgn\phi({\sf x}_j)$, which make
control dependencies in an error trace $P$ explicit. In order for
${\sf x}_i$ to take the value of ${\sf x}_j$, the outcomes of the
conditional statements preceding the assignment of ${\sf x}_j$ in
$P$ have to permit the assignment to be executed. 

Let $\stmt(n_j)$
be the statement assigning ${\sf x}_j$, and note that
control dependency coincides with our notion of a scope (as defined
in Definition \ref{def:scope}).
We define
\begin{equation}
    \guard(n_j)\defn\bigwedge\left\{
    \enc(n_i)\,\vert\,n_j \text{~is in scope of~}n_i
    \right\}\,.
\end{equation}

In order for the definition of ${\sf x}_j$ in $n_j$ to be reachable
along $P$, $\guard(n_j)$ needs to evaluate to \true.
Moreover, since trace $P$ does not traverse alternative branches,
the value of ${\sf x}_i$ is unknown if $\guard(n_j)$ does not hold.
Based on this insight, we define a \emph{control-sensitive} encoding
$\fsenc{P}$ as follows:
\begin{multline}
  \fsenc{n_i}\defn
  \left\{
  \begin{array}{ll}
    \guard(n_j)\Rightarrow({\sf x}_i={\sf x}_j) &
    \text{if}~\stmt(n_i)~\text{is}~{\sf x}_i\asgn\phi({\sf x}_j)\\
    & \quad\text{and}~n_j~\text{assigns}~{\sf x}_j\\
    \enc(n_i) & \text{if}~n_i~\text{is an assignment}\\
    \true & \text{if}~n_i~\text{is a condition}\\
  \end{array}
  \right.
  \label{eq:encode_phi}
\end{multline}

An inductive error invariant for the encoding 
$\fsenc{P}$ induces a control-sensitive slice
(cf.~Definition 4 of flow-sensitivity
and Theorem 6 in \cite{Christ:flowSensitive}):

\begin{definition}[Control-sensitive Slice]
  Let $P$ be an error trace
  for the assertion $\psi$. A (sound) slice $Q$ is 
  \emph{control-sensitive} if for every
  statement $\stmts_Q[k]=\stmts_P[i]$ and every 
  assumption $\stmts_P[j]$ such that 
  $\stmts_P[i]$ is in scope of $\stmts_P[j]$, 
  there is some prefix $\stmts_Q[1,h]$ of $\stmts_Q[1,k]$ (with $h<k$
  such that $\stmts_Q[h]$ precedes and $\stmts_Q[h+1]$
  succeeds or equals $\stmts_P[j]$
  in $P$) such that $\stmts_Q[1,h]$
  is an error trace for $\neg(\stmts_P[j])$.
  \label{def:fs_slice}
\end{definition}

Intuitively, the definition requires that $Q$
justifies that every branch containing a relevant statement
will be taken.

\begin{theorem}
  Let $P$ be a (concurrent) error trace for $\psi$ of length $k$ and let
  $I_0$, $I_1$ , \ldots ,$I_{k-1}$, $I_{k+1}$
  be error invariants (with $I_0=\true$ and
  $I_{k+1}=\false$) obtained from
  an inductive sequence interpolant
  for $\fsenc{n_1},\ldots,\fsenc{n_k},\psi$. Let $Q$ be the slice obtained
  from $P$ by removing each sub-path $P[i,j]$
  for which $I_{i-1}$ is inductive. Then $Q$
  is a sound control-sensitive slice for $P$.
  \label{thm:conc_loc}
\end{theorem}

Note that the interpolants in Theorem~\ref{thm:conc_loc} may
contain different versions of a variable {\sf x}, since
the encoding of $\phi$-nodes may refer to conditions in the ``past''.
This corresponds to \emph{history} or \emph{ghost} variables used
in Hoare logic and does not affect soundness.

\begin{example}
  Figure \ref{fig:cs_sensitive} shows the path $P$ from
  Example \ref{ex:cf_insensitive} sliced using a control-sensitive
  encoding $\fsenc{P}$ based on $\phi$-nodes. Note that
  the statements initializing {\sf bal} and {\sf amount},
  which guarantee that the \yes-branch is taken,
  are included in the slice.
\end{example}

\paragraph{Synchronization}
\label{sec:sync}
In the simple interleaving semantics deployed in this paper, locks
can be modeled using an integer variables $\ell$ and atomicity
constraints. Lock $\ell$ is available if its value is $0$.
Any other value $t$ indicates that the lock $\ell$ is held by thread
$t$. Let $n$ be a node of thread $t$ with a self-loop waiting for $(\ell=0)$
to become true, and $m$ its successor node assigning $t$ to $\ell$.
By constraining the execution such that 
no thread other than $t$ can execute between $n$ and $m$, we
guarantee that lock acquisition is performed atomically. Analogously,
a lock $\ell$ held by the current thread (guaranteed by condition $\ell=t$)
is released by the statement $\ell\asgn 0$.
Control-sensitive slices also take into account lock acquisition
statements, as relevant statements executed in a locked region
are in the scope of the corresponding condition $(\ell=0)$.

\paragraph{Hazards}
\label{subsec:data_dep}
A trace contains a \emph{data hazard} if its outcome depends on the
sequence (or timing) of concurrently executed statements. As
explained for the sub-trace in Figure \ref{fig:errinv_haz}
discussed in Section \ref{sec:introduction}, applying error
invariants in their original form \cite{ErmisSW12} to sequential
paths results in slices that ignore important
characteristics of concurrent traces. While $\fsenc{P}$ reflects
control-flow, it fails to capture \emph{data dependencies},
which are constraints arising from the flow of
data between statements \cite{muchnick}:
\begin{description}
\item[Read-after-write] 
  If statement $\stmt(n)$ writes a value read by
  statement $\stmt(m)$, then the two statements are 
  \emph{flow dependent}. 
\item[Write-after-read] An \emph{anti dependence} occurs when
  statement $\stmt(n)$ reads a value that is later
  updated (over-written) by $stmt(m)$.
\item[Write-after-write] An \emph{output dependence} exists
  if $\stmt(n)$ as well as $\stmt(m)$
  set the value of the same variable.
\end{description}

While this definition also applies to single threads, 
we concern ourselves exclusively with
\emph{inter-thread} data depen\-dencies. In a path $P$,
a data dependency
between different threads can indicate a conflicting
access (i.e., a \emph{race condition} or \emph{hazard}).

Unlike flow depen\-dence (which is taken into account by
$\enc(P)$ and $\fsenc{P}$, since the SSA form represents
use-definition pairs and therefore also flow dependence explicitly),
anti and output dependencies are not explicit in
the SSA-based encoding of $P$ used in
Sections \ref{sec:error_invariants} and
\ref{subsec:control_sensitive}. Similar to merge points in
sequential programs, inter-thread dependencies in $P$ give rise
to conflicting definitions of global variables.
The Concurrent SSA (CSSA) form of paths
presented in \cite{Wang2011SPA,SinhaW11}
introduces $\pi$-functions to resolve
dependencies between accesses to global variables
in different threads. 

To convert an error trace into CSSA form,
we introduce an
arbiter node before every read access to a global variable {\sf x} in
an error trace $P$ (analogously to the arbiter nodes for $\phi$-functions
in Section \ref{sec:preliminaries}). 
The arbiter node is annotated with a $\pi$-function that selects from all
definitions of the global variable {\sf x} in $P$ the most recent
definition:
\begin{description}
\item[$\pi$-functions] at $n\in N_t$ for a global variable {\sf x},
  take as parameters the subscripted variables representing
  definitions of {\sf x} in all threads.\footnote{
    As an optimization, only the \emph{last} definition
    of {\sf x} in thread $t$ before $n$ is added.
  }
\end{description}
\begin{figure}[t]
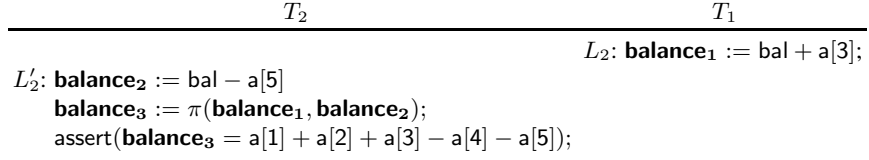
\centering
\begin{minipage}[t]{\textwidth}
\vspace{0pt}
\centering
{
\small
\begin{tabular}{ll@{\hspace{.45em}}ll}
\multicolumn{2}{c}{$T_2$} & \multicolumn{2}{c}{$T_1$}\\
\toprule
&& $L_2$: & $\boldsymbol{\mathsf{balance}_1}:=\mathsf{bal+\arrayvar[3]}$;\\
$L'_2$: & $\boldsymbol{\mathsf{balance}_2}:=\mathsf{bal-\arrayvar[5]}$\\
&${\boldsymbol{\mathsf{balance}_3}:=\pi(\boldsymbol{\mathsf{balance}_1},\boldsymbol{\mathsf{balance}_2});}$\\
&$\mathsf{assert (\boldsymbol{\mathsf{balance}_3}=a[1]+a[2]+a[3]-a[4]-a[5]);}$
\end{tabular}
}
\caption{Part of a path with hazard and $\pi$-node\label{fig:pi}}
\vspace{-1.7em}
\end{minipage}
\end{figure}

Figure~\ref{fig:pi} shows a simplified suffix of the trace in
Figure~\ref{fig:errinv_haz}. The simplified trace consists of two
threads with a $\pi$-node (arbitrating between the definitions ${\sf
  balance}_1$ and ${\sf balance}_2$) inserted before an assertion
$\psi$ that states the expected outcome.  Note that unlike the
degenerate $\phi$-functions used in Section
\ref{subsec:control_sensitive}, a $\pi$-function for {\sf x} has as
many parameters as there are definitions of {\sf x} in $P$.

To encode WAR and WAW dependencies,
we introduce an irreflexive, transitive, and anti-symmetric
relation $\hb{n_i}{n_j}$ which indicates that node $n_i$
is executed before node $n_j$. This happens-before relation
enables us to encode the edges of a program trace, reflecting
the program order and the schedule.

In addition, $\rdvar{{\sf x}}{n_i}$ and $\wrvar{{\sf x}}{n_j}$
indicate that {\sf x} is read at node $n_i$ and written at node $n_j$.
These primitives allow for an explicit encoding of data dependencies:
\begin{equation}
  \begin{array}{rcl}
    \wrvar{{\sf x}}{n_i}\wedge\hb{n_i}{n_j}\wedge\rdvar{{\sf x}}{n_j}
    & \Leftrightarrow & \raw{\sf x}{n_i}{n_j} \\
    \rdvar{{\sf x}}{n_i}\wedge\hb{n_i}{n_j}\wedge\wrvar{{\sf x}}{n_j}
    & \Leftrightarrow & \war{\sf x}{n_i}{n_j} \\
    \wrvar{{\sf x}}{n_i}\wedge\hb{n_i}{n_j}\wedge\wrvar{{\sf x}}{n_j}
    & \Leftrightarrow & \waw{\sf x}{n_i}{n_j}
  \end{array}
\end{equation}

The hazard-sensitive encoding presented below incorporates
data dependen\-cies into the encoding of a trace. The encoding is derived
directly from a program path $P$, taking advantage of the information
encoded in the edges. Assignments (without $\pi$-functions) are encoded
as follows:
\begin{multline}
  \hsenc{n_i}\defn
  \left\{
  \begin{array}{ll}
    \wrvar{\sf x}{n_i}\wedge\enc(n_i) &
    \text{if}~n_i~\text{writes global var.~{\sf x}}\\
    \rdvar{\sf x}{n_i}\wedge\enc(n_i) &
    \text{if}~n_i~\text{reads global var.~{\sf x}}\\
    \enc(n_i) & \text{otherwise}
  \end{array}
  \right.
\end{multline}

Nodes $n_i$ with $\pi$-functions incorporate happens-before
information. Let $n_i$ be a $\pi$-node assigning ${\sf x}_i$,
let $n_j$ be an assignment to ${\sf x}_j$ and 
the last node before $n_i$ in $P$ updating
the global variable {\sf x}. Then $\hsenc{n_i}$ is:
\begin{equation}
  \rdvar{\sf x}{n_i}\wedge
  \left(\mathrm{DEP}(n_i, n_j) \Rightarrow ({\sf x}_i={\sf x}_j)\right)
  \label{eq:pi-encoding1}
\end{equation}
where $\mathrm{DEP}(n_i,n_j)$ is the following condition:
\begin{equation}
    \raw{\sf x}{n_j}{n_i}\wedge
    \bigwedge_{\footnotesize\begin{array}{c}m\in\{n\in P\,\vert\,
      \wrvar{\sf x}{n}\}\cr m\neq
        n_j\end{array}}
    (\waw{\sf x}{m}{n_j}\vee\war{\sf x}{n_i}{m})
  \label{eq:pi-encoding2}
\end{equation}

Intuitively, $\mathrm{DEP}(n_i,n_j)$ states that
${\sf x}_j$ is written before ${\sf x}_i$ is read,
and no other definition of {\sf x} interferes.

Finally, edges are encoded as happens-before relations:
\begin{equation}
  \hsenc{\langle n_i,n_{i+1}\rangle}\defn\hb{n_i}{n_{i+1}}
\end{equation}

Given a path $P\defn n_1,\langle n_1,n_2\rangle, n_2,
\ldots, \langle n_{k-1},n_k\rangle, n_k$, applying sequence interpolation
to the formulas
$\hsenc{n_1}$, $\hsenc{\langle n_1,n_2\rangle}$, $\hsenc{n_2}$, \ldots,
$\hsenc{\langle n_{k-1},n_k\rangle}$, $\hsenc{n_k}$, $\psi$
yields a sequence $\In_1, \Out_1,\ldots,\In_k,\Out_k$ 
of formulas such that
\begin{displaymath}
  \In_i\wedge\hsenc{n_i}\Rightarrow\Out_i ~\text{and}~
  \Out_i\wedge\hsenc{\langle n_i,n_{i+1}\rangle}\Rightarrow\In_{i+1}\,.
\end{displaymath}

Unlike before, $\In_i$ and $\Out_i$ propagate facts about states
as well as execution order. We can slice
sub-path $[n_i,n_j]$ if $\In_i\Rightarrow\Out_j$, sub-path $(n_i,n_j)$ if
$\Out_i\Rightarrow\In_j$, sub-path $[n_i,n_j)$ if
  $\In_i\Rightarrow\In_j$, and sub-path $(n_i,n_j]$ if
$\Out_i\Rightarrow\Out_j$.
The resulting sliced path $Q$ corresponds to a sequence
of statements $\stmts_Q$
and a set of edges $\project{E}{Q}$ representing
context switches and program order constraints relevant to the error.

\begin{definition}[Hazard-sensitive slice]
  Given an error trace $P$, a 
  (sound) slice $Q$ is \emph{hazard-sensitive} if for every
  statement $\stmts_Q[k]=\stmts_P[j]$ and 
  statement $\stmts_P[i]$ such that there is an inter-thread
  data dependency between $\stmts_P[i]$ and $\stmts_P[j]$, there is an $h$
  such that $\stmts_Q[h]=\stmts_P[j]$.
  \label{def:hazard_slice}
\end{definition}

\begin{theorem}
  Let $P$ be a concurrent error trace and let
  $Q$ be the slice obtained
  from $P$ as explained above.
  Then $Q$ is a sound hazard-sensitive slice of $P$.
  \label{thm:hazard_loc}
\end{theorem}

\begin{example}
  Consider the path in Figure \ref{fig:pi}. A hazard-insensitive
  slice would contain the statement at node $L'_2$
  but not the statement at node $L'_2$ (as explained
  in Section \ref{sec:introduction}) since $L_2$ has no
  influence on the state after $L'_2$. 
  Encoding (\ref{eq:pi-encoding1}) and (\ref{eq:pi-encoding2})
  of the $\pi$-node require the interpolant before
  the $\pi$-node to imply
  $\waw{\sf balance}{L_2}{L'_2}$,
  and consequently $\wrvar{\sf balance}{L_2}$,
  $\wrvar{\sf balance}{L'_2}$, and $\hb{L_2}{L'_2}$
  (as indicated in Figure \ref{fig:errinv_haz}). 
  Nodes $L_2$ and $L'_2$ as well as the edge
  $\langle L_2,L'_2\rangle$ are included in the resulting slice.
\end{example}

\subsection{Fine-Tuning Explanations}
\label{subsec:conc_enc} 
The encodings presented in Section~\ref{subsec:concurrency} can be combined 
in a straightforward manner, providing us with a choice of control
WAR , and WAW dependencies reflected by the resulting explanation.
Control-flow or hazard-sensitivity can be added (or removed) by (dis-)regarding
$\pi$-nodes and $\phi$-nodes in $P$.
Control-flow dependency can be incorporated
into $\pi$-nodes in Equation (\ref{eq:pi-encoding1}) by prefixing
the assignment ${\sf x}_i={\sf x}_j$ with the guard
of the definition of ${\sf x}_j$ at node $n_j$:
 $ \guard(n_j)\Rightarrow
  \left(\mathrm{DEP}(n_i, n_j) \Rightarrow ({\sf x}_i={\sf x}_j)\right)$,
similar to the guard in the definition of $\fsenc{n_i}$
in Encoding (\ref{eq:encode_phi}). Moreover,
Encoding (\ref{eq:pi-encoding1}) can be made
insensitive to WAR dependencies by restricting $m$ to
predecessors of $n_i$ and by dropping the disjunct $\war{\sf x}{n_i}{m}$
from (\ref{eq:pi-encoding2}) (and similarly for WAW dependencies).
Note that flow dependency has a special role, since 
use-definition chains are explicit in the SSA representation.

\noindent\parbox{.6\textwidth}{
The partial order given by the subset relation $\subseteq$
over the power-set of the remaining dependencies
$\{\sf{cs}, \sf{war}, \sf{waw}\}$
reflects possible levels of detail of explanations, as illustrated
by the Hasse diagram to the right.
As indicated in the diagram, the configuration
$\emptyset$ corresponds to the basic approach presented in
\cite{ErmisSW12,Murali:hybrid}, whereas $\{\sf{cs}\}$
represents control-flow sensitive approach.
}
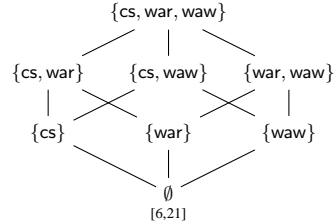
\raisebox{-3.5em}{
  \begin{tikzpicture}[scale=.8,every node/.style={scale=0.8}]
    \node (n0) at (3,3) {\small $\{{\sf cs}, {\sf war}, {\sf waw}\}$};
    \node (n1) at (1,2) {\small $\{{\sf cs}, {\sf war}\}$};
    \node (n2) at (3,2) {\small $\{{\sf cs}, {\sf waw}\}$};
    \node (n3) at (5,2) {\small $\{{\sf war}, {\sf waw}\}$};
    \draw (n0)--(n1);
    \draw (n0)--(n2);
    \draw (n0)--(n3);
    \node (n5) at (1,1) {\small $\{{\sf cs}\}$};
    \node (n6) at (3,1) {\small $\{{\sf war}\}$};
    \node (n7) at (5,1) {\small $\{{\sf waw}\}$};
    \draw (n1)--(n5);
    \draw (n2)--(n5);
    \draw (n1)--(n6);
    \draw (n3)--(n6);
    \draw (n2)--(n7);
    \draw (n3)--(n7);
    \node (n11) at (3,0) {\small $\emptyset$};
    \draw (n11)--(n5);
    \draw (n11)--(n6);
    \draw (n11)--(n7);
    \node (n12) at (3,-.35)
          {\scriptsize\cite{ErmisSW12,Murali:hybrid}};
  \end{tikzpicture}
}

While we see interpolants as an inherent part of the explanation, the level
of detail provided by these annotations cannot be related or
formalized as easily as it is the case for dependencies: changing
the underlying encoding typically has an unpredictable effect
on the structure and strength of interpolants~\cite{dpwk2010,impact}.
%

\section{Experiments}
\label{sec:experiments}
We implemented our approach as an extension of the directed
testing tool {\ConCrest}~\cite{FarzanHRV13}.
We generate error traces of concurrent programs and then
produce slices as described in Section~\ref{sec:fault-localization-paths}. 
While all slices provided by our tool are sound in the sense of
Definition~\ref{def:sound_slice}, the level of detail might not
be sufficient to reflect the underlying bug: for example,
the hazard-sensitive slice for the \textsf{account} benchmark readily
reveals the atomicity violation. Therefore, it is not necessary to
compute a more detailed control-sensitive slice.

The results from
Section~\ref{subsec:conc_enc} enable the developer to
gradually increase the detail in an iterative manner until
the bug can be understood.
This section provides an empirical evaluation of the
size and accuracy of slices with varying levels of
detail. %

\paragraph{Effectiveness of the Method}
To evaluate our method, we applied it on a collection of faulty C
programs to show how effective the different dependency encodings are
at revealing different types of concurrency bugs. We used four
different encodings to track data and control dependencies:
\textbf{hs} refers to hazard-sensitive encoding for tracking inter-thread data
dependencies, \textbf{cs} refers to
control-sensitive encoding for tracking control dependencies, and \textbf{ds} denotes
the basic encoding $\enc_P$ of Section
\ref{sec:fault-localization-paths}.
The symbol ``+'' indicates combinations of encodings.

Our definition of whether the bug was captured depends on the type of bug.
For data race bugs, we required that the slice reflecting the bug
contains both conflicting accesses.  For atomicity violations, a slice
reflecting the bug contains conflicting statements from another thread
interrupting the desired atomic region.  For order violations, a slice
reflecting the bug contains conflicting statements in the problematic
order. 

Table~\ref{tbl:benchmarks} summarizes our empirical results.  
The benchmarks in this table are
classified into two groups.  
The first group consists of 33 multithreaded C programs taken
from~\cite{Khoshnood2015}.\footnote{%
  {\ConCrest}'s search heuristic failed to generate an error trace for the  
  \textsf{fibbench\_longer}, a variant of \textsf{fibbench} with
  larger parameters. We emphasize that this failure is related
  to the generation of traces rather than slicing.}
These programs capture the essence of concurrency bugs reported in
various versions of open source applications such as Mozilla, Apache,
and GCC.  The \textsf{apache2} and \textsf{bluetooth} benchmarks in
the second group are simplified versions of applications taken from
\cite{FarzanHRV13}.  The \textsf{pool-simple-2} benchmark is a lock-free
concurrent data structure with a linearizability bug. We discuss this
benchmark in depth in App.~\ref{sec:case_study_pool}.
The remaining two benchmarks in the second group are variants of the 
program discussed in Section~\ref{sec:introduction}.
For each benchmark program, the name, the number of lines of code
(LOC), the number of threads, and the type of bug are listed in
Table~\ref{tbl:benchmarks}.
The number of error traces (\#T) per benchmark varies due to specific 
assertions and {\ConCrest}'s
ability to produce error traces. They do not reflect any
preselection of traces. In total, {\ConCrest} generated 90 error
traces from the 38 programs all of which we considered in our
evaluation.

We use $\checkmark$ to indicate that the explanations obtained using the
corresponding encoding capture the bug, and -- if the bug was not captured.
By manually inspecting the slices we found that for all but two
benchmarks, tracking all dependencies \textbf{ds+cs+hs} yields
explanations that capture the corresponding concurrency bug.
For most benchmarks there exists at least one additional encoding
which provides smaller slices that still reveal the
bug. This encoding is usually \textbf{hs} (68\%) or \textbf{cs}
(50\%) depending on the nature of the bug and the assertions.
Interestingly, our analysis revealed that \textsf{boop},
\textsf{freebsd\_auditarg} and \textsf{gcc-java-25530}
from~\cite{Khoshnood2015} contain sequential bugs already reflected
in a \textbf{ds}-slice rather than
concurrency bugs (even though in~\cite{Khoshnood2015} they are classified
as concurrency bugs).

In two of the three error traces of \textsf{freebsd\_auditarg}
the bug is triggered by non-interleaved executions of the
threads. For these traces, any 
encoding yields an adequate explanation. In one error trace,
however, the bug is triggered by an interference between two threads,
which is only reflected by the encodings \textbf{ds+hs} and \textbf{ds+cs+hs}.

Only the programs \textsf{hash\_table}, \textsf{ms\_queue02}, and
\textsf{list\_seq}, which contain bugs in intricate
concurrent data structures, 
require the full \textbf{ds+cs+hs} encoding. 

Only for the two benchmarks \textsf{apache-25520} and
\textsf{cherokee\_01} the slices produced by our method failed
to reveal the bugs. The problem is that 
the root cause of the assertion violation is that a specific branch of
a conditional statement is not taken during the execution. 
Slices of single error traces cannot reveal the non-occurrence of an event as
the cause for failure. Therefore, we plan to analyze merged
error traces in future work.

\paragraph{Running times} 
The generation of the slices takes an
average of 2.43s ($\sigma=11.02s$) across all encodings
with a maximum of 168.8s. As expected, the running times increase with
the amount of detail captured by the encoding. Generating a
\textbf{ds} explanation takes 0.43s on average ($\sigma=0.18s$)
whereas a \textbf{ds+cs+hs} explanation takes 7.3s ($\sigma=21.25s$).

\paragraph{Quantitative Evaluation}
Table~\ref{tbl:benchmarks} shows the effect of tracking different
dependencies on the size of the slices.  $\mu$ refers to average
percentage reduction as the quotient of the number of remaining and
original instructions, so smaller numbers mean smaller slices.  As
expected, increasing the sensitivity of the algorithm by tracking more
dependencies leads to smaller reductions.
However, as we have seen previously, the hazard-sensitive
explanations (\textbf{ds+hs}), which capture the concurrency bugs in
68\% of the benchmarks, on average contain 35\% of the original
instructions and 54\% of the original variables.  We gained the
maximum reduction with the encoding
(\textbf{ds}), however the resulting explanations reflected the
concurrency bugs in only 23\% of the benchmarks.  The amount of
reduction differs across benchmarks with a maximum of 93\% for
the \textsf{apache2} benchmark program.
Slices which are hazard- but not control-flow
sensitive tend to be much smaller than slices which are control-flow
sensitive, but not data-hazard sensitive.
\setlength{\textfloatsep}{0.1cm}
\newcolumntype{C}{>{\centering\arraybackslash}X}
\newcolumntype{R}{>{\raggedleft\arraybackslash}X}
\newcolumntype{L}{>{\raggedright\arraybackslash}X}
\begin{table}[t!]
\centering
\tiny

\newcommand{\myspace}{\hspace{.5em}}
 \begin{tabularx}{\textwidth}{X|r|r@{\hspace{3pt}}r|r|c|c|rrrrc|rrrrc|rrrrc|rrrrc}

  	\toprule
		\textbf{Benchmark} & 
		\textbf{\#T} & 
		\multicolumn{2}{c|}{LOC} & 
		\multicolumn{1}{c|}{AIT} & 
		Threads & 
		Bugs & 
		\multicolumn{5}{c}{\textbf{ds+cs+hs}} & 
		\multicolumn{5}{c}{\textbf{ds+hs}} &
		\multicolumn{5}{c}{\textbf{ds+cs}} & 
		\multicolumn{5}{c}{\textbf{ds}} \\
		
  	& 
  	& 
  	& 
  	& 
  	& 
  	& 
  	& 
  	\multicolumn{2}{@{}c@{}}{S[\%]} &  
  	\multicolumn{2}{@{}c@{}}{V[\%]} & 
  	&
  	\multicolumn{2}{@{}c@{}}{S[\%]} &  
  	\multicolumn{2}{@{}c@{}}{V[\%]} & 
  	&
  	\multicolumn{2}{@{}c@{}}{S[\%]} &  
  	\multicolumn{2}{@{}c@{}}{V[\%]} & 
  	&
  	\multicolumn{2}{@{}c@{}}{S[\%]} &  
  	\multicolumn{2}{@{}c@{}}{V[\%]} &
  	\\
  	
  	& 
  	& 
  	& 
  	& 
  	& 
  	& 
  	& 
  	\multicolumn{1}{c}{$\mu$} & 
  	\multicolumn{1}{c}{$\sigma$} & 
  	\multicolumn{1}{c}{$\mu$} & 
  	\multicolumn{1}{c}{$\sigma$} & 
  	\multicolumn{1}{c|}{RB} &
  	\multicolumn{1}{c}{$\mu$} & 
  	\multicolumn{1}{c}{$\sigma$} & 
  	\multicolumn{1}{c}{$\mu$} & 
  	\multicolumn{1}{c}{$\sigma$} & 
  	\multicolumn{1}{c|}{RB} &
  	\multicolumn{1}{c}{$\mu$} & 
  	\multicolumn{1}{c}{$\sigma$} & 
  	\multicolumn{1}{c}{$\mu$} & 
  	\multicolumn{1}{c}{$\sigma$} & 
  	\multicolumn{1}{c|}{RB} &
  	\multicolumn{1}{c}{$\mu$} & 
  	\multicolumn{1}{c}{$\sigma$} & 
  	\multicolumn{1}{c}{$\mu$} & 
  	\multicolumn{1}{c}{$\sigma$} &
  	RB\\
  	\midrule

account                   & 3  & 43  & (58)  & 51.7  & 4 & AV & \textbf{62}  & \textbf{12} & \textbf{77}  & \textbf{6} & $\checkmark$ & \textbf{42}  & \textbf{10} & \textbf{68}  & \textbf{6} & $\checkmark$ & 43          & 8           & 68           & 6           & --           & 29          & 5          & 59           & 5           & -- \\

apache-21287              & 2  & 30  & (79)  & 43    & 3 & AV & \textbf{72}  & \textbf{0}  & \textbf{87}  & \textbf{0} & $\checkmark$ & 28           & 0           & 53           & 0          & --           & \textbf{51} & \textbf{0}  & \textbf{87}  & \textbf{0}  & $\checkmark$ & 9           & 0          & 40           & 0           & -- \\

apache-25520              & 1  & 88  & (192) & 34    & 3 & AV & 38           &  --         & 50           & --         & --           & 9            &  --         & 33           & --         & --           & 26          &  --         & 50           & --          & --                 & 9           &  --        & 33           & --          & --                 \\

barrier\_vf\_false        & 12 & 57  & (85)  & 27    & 4 & AV & \textbf{70}  & \textbf{0}  & \textbf{80}  & \textbf{0} & $\checkmark$ & 19           & 0           & 40           & 0          & --           & \textbf{67} & \textbf{0}  & \textbf{80}  & \textbf{0}  & $\checkmark$       & 15          & 0          & 40           & 0           & --                 \\

boop                      & 1  & 58  & (98)  & 40    & 3 & SB & \textbf{38}  &  --         & \textbf{47}  & --         & $\checkmark$ & \textbf{30}  &  --         & \textbf{40}  & --         & $\checkmark$ & \textbf{35} &  --         & \textbf{47}  & --          & $\checkmark$       & \textbf{28} &  --        & \textbf{40}  & --          & $\checkmark$       \\

cherokee\_01              & 1  & 88  & (188) & 28    & 3 & AV & 46           &  --         & 60           & --         & --           & 11           &  --         & 40           & --         & --           & 32          &  --         & 60           & --          & --                 & 11          &  --        & 40           & --          & --                 \\

counter\_seq              & 1  & 28  & (41)  & 29    & 3 & DR & \textbf{72}  &  --         & \textbf{90}  & --         & $\checkmark$ & \textbf{38}  &  --         & \textbf{70}  & --         & $\checkmark$ & 52          &  --         & 80           & --          & --                 & 31          &  --        & 60           & --          & --                 \\

fibbench                  & 2  & 34  & (47)  & 34    & 3 & AV & \textbf{94}  & \textbf{3}  & \textbf{97}  & \textbf{3} & $\checkmark$ & \textbf{94}  & \textbf{3}  & \textbf{97}  & \textbf{3} & $\checkmark$ & \textbf{88} & \textbf{3}  & \textbf{97}  & \textbf{3}  & $\checkmark$       & \textbf{88} & \textbf{3} & \textbf{97}  & \textbf{3}  & $\checkmark$       \\

freebsd\_auditarg         & 3  & 52  & (104) & 37    & 4 & SB & \textbf{67}  & \textbf{7}  & \textbf{86}  & \textbf{0} & $\checkmark$ & \textbf{32}  & \textbf{5}  & \textbf{64}  & \textbf{0} & $\checkmark$ & \textbf{57} & \textbf{10} & \textbf{79}  & \textbf{10} & $\checkmark$ (2/3) & \textbf{30} & \textbf{8} & \textbf{57}  & \textbf{10} & $\checkmark$ (2/3) \\

gcc-java-25530            & 2  & 36  & (86)  & 17    & 3 & SB & \textbf{35}  & \textbf{0}  & \textbf{40}  & \textbf{0} & $\checkmark$ & \textbf{35}  & \textbf{0}  & \textbf{40}  & \textbf{0} & $\checkmark$ & \textbf{24} & \textbf{0}  & \textbf{40}  & \textbf{0}  & $\checkmark$       & \textbf{24} & \textbf{0} & \textbf{40}  & \textbf{0}  & $\checkmark$       \\

gcc-libstdc++-3584        & 1  & 40  & (104) & 37    & 3 & AV & \textbf{62}  &  --         & \textbf{79}  & --         & $\checkmark$ & \textbf{35}  &  --         & \textbf{64}  & --         & $\checkmark$ & 46          &  --         & 71           & --          & --                 & 30          &  --        & 57           & --          & --                 \\

gcc-libstdc++-21334       & 1  & 36  & (86)  & 27    & 3 & OV & \textbf{63}  &  --         & \textbf{78}  & --         & $\checkmark$ & \textbf{22}  &  --         & \textbf{33}  & --         & $\checkmark$ & 48          &  --         & 78           & --          & --                 & 15          &  --        & 33           & --          & --                 \\

gcc-libstdc++-40518       & 2  & 40  & (104) & 23    & 3 & AV & \textbf{43}  & \textbf{0}  & \textbf{56}  & \textbf{0} & $\checkmark$ & 30           & 0           & 56           & 0          & --           & \textbf{39} & \textbf{0}  & \textbf{56}  & \textbf{0}  & $\checkmark$       & 22          & 0          & 56           & 0           & --                 \\

glib-512624\_02           & 2  & 50  & (94)  & 27.5  & 3 & AV & \textbf{84}  & \textbf{2}  & \textbf{100} & \textbf{0} & $\checkmark$ & \textbf{47}  & \textbf{3}  & \textbf{80}  & \textbf{0} & $\checkmark$ & 60          & 4           & 85           & 5           & --                 & 38          & 5          & 65           & 5           & --                 \\

hash\_table               & 1  & 51  & (114) & 69    & 3 & AV & \textbf{41}  &  --         & \textbf{61}  & --         & $\checkmark$ & 4            &  --         & 21           & --         & --           & 29          &  --         & 54           & --          & --                 & 4           &  --        & 21           & --          & --                 \\

jetty-1187                & 1  & 24  & (98)  & 26    & 3 & AV & \textbf{81}  &  --         & \textbf{100} & --         & $\checkmark$ & \textbf{35}  &  --         & \textbf{78}  & --         & $\checkmark$ & 58          &  --         & 89           & --          & --                 & 27          &  --        & 67           & --          & --                 \\

lazy01\_false             & 2  & 39  & (55)  & 23    & 4 & OV & \textbf{91}  & \textbf{0}  & \textbf{100} & \textbf{0} & $\checkmark$ & \textbf{65}  & \textbf{0}  & \textbf{100} & \textbf{0} & $\checkmark$ & \textbf{87} & \textbf{0}  & \textbf{100} & \textbf{0}  & $\checkmark$       & \textbf{61} & \textbf{0} & \textbf{100} & \textbf{0}  & $\checkmark$       \\

lineEq\_2t\_01            & 1  & 35  & (58)  & 52    & 3 & AV & \textbf{69}  &  --         & \textbf{81}  & --         & $\checkmark$ & \textbf{46}  &  --         & \textbf{71}  & --         & $\checkmark$ & \textbf{52} &  --         & \textbf{76}  & --          & $\checkmark$       & \textbf{37} &  --        & \textbf{67}  & --          & $\checkmark$       \\

linux-iio                 & 1  & 54  & (87)  & 55    & 3 & DR & \textbf{40}  &  --         & \textbf{59}  & --         & $\checkmark$ & \textbf{20}  &  --         & \textbf{50}  & --         & $\checkmark$ & 27          &  --         & 41           & --          & --                 & 16          &  --        & 32           & --          & --                 \\

linux-tg3                 & 1  & 93  & (115) & 167   & 3 & DR & \textbf{19}  &  --         & \textbf{38}  & --         & $\checkmark$ & \textbf{13}  &  --         & \textbf{36}  & --         & $\checkmark$ & \textbf{8}  &  --         & \textbf{11}  & --          & $\checkmark$       & 2           &  --        & 9            & --          & --                 \\

list\_seq                 & 1  & 59  & (122) & 53    & 3 & AV & \textbf{58}  &  --         & \textbf{95}  & --         & $\checkmark$ & 6            &  --         & 30           & --         & --           & 40          &  --         & 75           & --          & --                 & 6           &  --        & 30           & --          & --                 \\

llvm-8441                 & 2  & 149 & (244) & 32.5  & 3 & AV & \textbf{74}  & \textbf{4}  & \textbf{92}  & \textbf{0} & $\checkmark$ & \textbf{18}  & \textbf{0}  & \textbf{33}  & \textbf{0} & $\checkmark$ & 55          & 7           & 83           & 8           & --                 & 12          & 0          & 33           & 0           & --                 \\

mozilla-61369             & 1  & 19  & (68)  & 6     & 1 & OV & \textbf{67}  &  --         & \textbf{100} & --         & $\checkmark$ & \textbf{67}  &  --         & \textbf{100} & --         & $\checkmark$ & \textbf{67} &  --         & \textbf{100} & --          & $\checkmark$       & \textbf{67} &  --        & \textbf{100} & --          & $\checkmark$       \\

ms\_queue02               & 1  & 67  & (97)  & 66    & 3 & AV & \textbf{44}  &  --         & \textbf{52}  & --         & $\checkmark$ & 5            &  --         & 20           & --         & --           & 35          &  --         & 48           & --          & --                 & 5           &  --        & 20           & --          & --                 \\

mysql5                    & 1  & 21  & (27)  & 28    & 3 & AV & \textbf{82}  &  --         & \textbf{89}  & --         & $\checkmark$ & \textbf{46}  &  --         & \textbf{67}  & --         & $\checkmark$ & 46          &  --         & 89           & --          & --                 & 25          &  --        & 67           & --          & --                 \\

mysql-644                 & 1  & 68  & (165) & 16    & 3 & AV & \textbf{38}  &  --         & \textbf{33}  & --         & $\checkmark$ & \textbf{38}  &  --         & \textbf{33}  & --         & $\checkmark$ & 25          &  --         & 33           & --          & --                 & 25          &  --        & 33           & --          & --                 \\

mysql-3596                & 1  & 30  & (83)  & 6     & 3 & DR & \textbf{100} &  --         & \textbf{100} & --         & $\checkmark$ & \textbf{100} &  --         & \textbf{100} & --         & $\checkmark$ & \textbf{67} &  --         & \textbf{100} & --          & $\checkmark$       & \textbf{67} &  --        & \textbf{100} & --          & $\checkmark$       \\

mysql-12848               & 1  & 51  & (142) & 14    & 2 & AV & \textbf{71}  &  --         & \textbf{67}  & --         & $\checkmark$ & 43           &  --         & 50           & --         & --           & \textbf{50} &  --         & \textbf{67}  & --          & $\checkmark$       & 29          &  --        & 50           & --          & --                 \\

read\_write\_false        & 1  & 78  & (140) & 58    & 5 & AV & \textbf{17}  &  --         & \textbf{27}  & --         & $\checkmark$ & \textbf{17}  &  --         & \textbf{27}  & --         & $\checkmark$ & \textbf{17} &  --         & \textbf{27}  & --          & $\checkmark$       & \textbf{17} &  --        & \textbf{27}  & --          & $\checkmark$       \\

reorder2\_false           & 8  & 50  & (105) & 10.5  & 5 & AV & \textbf{86}  & \textbf{14} & \textbf{100} & \textbf{0} & $\checkmark$ & \textbf{86}  & \textbf{14} & \textbf{100} & \textbf{0} & $\checkmark$ & \textbf{62} & \textbf{8}  & \textbf{100} & \textbf{0}  & $\checkmark$       & \textbf{62} & \textbf{8} & \textbf{100} & \textbf{0}  & $\checkmark$       \\

testconc02                & 1  & 15  & (19)  & 9     & 2 & AV & \textbf{89}  &  --         & \textbf{100} & --         & $\checkmark$ & \textbf{89}  &  --         & \textbf{100} & --         & $\checkmark$ & 56          &  --         & 100          & --          & --                 & 56          &  --        & 100          & --          & --                 \\

transmission-1.42         & 1  & 25  & (78)  & 5     & 3 & DR & \textbf{100} &  --         & \textbf{100} & --         & $\checkmark$ & \textbf{100} &  --         & \textbf{100} & --         & $\checkmark$ & \textbf{80} &  --         & \textbf{100} & --          & $\checkmark$       & \textbf{80} &  --        & \textbf{100} & --          & $\checkmark$       \\

VectPrime02               & 1  & 97  & (183) & 115   & 3 & AV & \textbf{25}  &  --         & \textbf{68}  & --         & $\checkmark$ & \textbf{9}   &  --         & \textbf{45}  & --         & $\checkmark$ & 18          &  --         & 59           & --          & --                 & 7           &  --        & 36           & --          & --                 \\

\midrule

apache2                   & 8  & 719 & (--)  & 235.5 & 3 & AV & \textbf{8}   & \textbf{2}  & \textbf{9}   & \textbf{2} & $\checkmark$ & 1            & 0           & 1            & 0          & --           & \textbf{7}  & \textbf{2}  & \textbf{9}   & \textbf{2}  & $\checkmark$       & 1           & 0          & 1           & 0            & $\checkmark$       \\

bankaccount-lock-for-loop & 5  & 103 & (--)  & 247   & 3 & AV & \textbf{46}  & \textbf{2}  & \textbf{44}  & \textbf{2} & $\checkmark$ & \textbf{12}  & \textbf{1}  & \textbf{30}  & \textbf{2} & $\checkmark$ & 40          & 2           & 42           & 2           & --                 & 9           & 1          & 23          & 3            & --                 \\

bankaccount-simple-lock   & 2  & 50  & (--)  & 45    & 3 & AV & \textbf{71}  & \textbf{0}  & \textbf{80}  & \textbf{0} & $\checkmark$ & \textbf{31}  & \textbf{0}  & \textbf{60}  & \textbf{0} & $\checkmark$ & 62          & 0           & 73           & 0           & --                 & 24          & 0          & 53          & 0            & --                 \\

bluetooth                 & 5  & 87  & (--)  & 35.8  & 3 & AV & \textbf{42}  & \textbf{0}  & \textbf{63}  & \textbf{0} & $\checkmark$ & 14           & 0           & 31           & 0          & --           & \textbf{36} & \textbf{0}  & \textbf{63}  & \textbf{0}  & $\checkmark$       & 11          & 0          & 31          & 0            & --                 \\

pool-simple-2             & 8  & 298 & (--)  & 885.5 & 3 & LV & \textbf{30}  & \textbf{1}  & \textbf{58}  & \textbf{2} & $\checkmark$ & 0            & 0           & 2            & 0          & --           & \textbf{29} & \textbf{1}  & \textbf{56}  & \textbf{2}  & $\checkmark$       & 0           & 0          & 2           & 0            & --                 \\

\bottomrule
Total                     & 90 &     &       &       &   &    & 58.8         &             & 72           &            & 88           & 35           &             & 54           &            & 47           & 45          &             & 67.7         &             & 61                 & 27          &            & 50.5        &              & 22                 \\
\bottomrule
  \end{tabularx}\\[2pt]
	\scriptsize
	  \begin{tabularx}{\textwidth}{@{}ll@{\hspace{1em}}ll@{\hspace{1em}}ll@{}}
    \textbf{\#T}: & No. of Traces in Benchmark &
    \textbf{LOC}: & Lines of Code$^a$ & AIT: & Average No. of Instructions in a Trace \\ 
    \textbf{ds}: & Basic Encoding & \textbf{cs}: & Control-Sensitive Encoding &
		\textbf{hs}: & Hazard-Sensitive Encoding \\ 
    S: & Slice Size / Trace Size &
    V: & \multicolumn{3}{@{}l} {No. of Variables in Slice / No. of Variables in Trace}  \\ 
    $\mu$: & Average &
    $\sigma$: & \multicolumn{3}{@{}l} {Standard Deviation} \\ 
    \textbf{RB}: & Reflects Concurrency Bug &
		AV: & Atomicity Violation & SB: & Sequential Bug \\ 
    DR: & Data Race & OV: & Order Violation &
    LV: & Linearizability Violation\\    
  \end{tabularx}\\
  {
  	\tiny
  	\begin{tabularx}{\textwidth}{lL}
  	$^a$ & LOC excluding comments and blank lines; LOC in
  parentheses are as stated in~\cite{Khoshnood2015}.
  	\end{tabularx}
  }
  \caption{Experimental comparison of sensitivity-configurations for
  slicing\label{tbl:benchmarks}}
\end{table}

%
%
%
%
%
%
%

\section{Related Work}
\label{sec:related-work}
The original work on
error invariants~\cite{ErmisSW12,Christ:flowSensitive} is
discussed in Sections \ref{sec:preliminaries} and
\ref{sec:fault-localization-paths}. 
Murali et al.~\cite{Murali:hybrid} relate error
invariants to unsatisfiable cores and consistency-based
diagnosis. The latter is also implemented in
\Concbugassist~\cite{Khoshnood2015}, a repair tool for concurrent
programs, and \Bugassist~\cite{Rupak:bugassist} for the
diagnosis of sequential bugs.
Both \Bugassist and \Concbugassist take into account multiple
traces simultaneously and can yield better accuracy in certain cases
(e.g., benchmarks \textsf{apache-25520} and \textsf{cherokee\_01} in
Section \ref{sec:experiments}). Neither
\cite{Rupak:bugassist,Khoshnood2015} nor \cite{Murali:hybrid}
report branch conditions (or statements explaining why they
hold).
On the benchmarks from~\cite{Khoshnood2015}, we found that
\Concbugassist yields similar reduction ratios as our tool using the
{\bf hs+ds} encoding. The dependency of \Concbugassist on a bounded
model checker for the constraint generation entails scalability
issues: even on a simplified version of \textsf{pool\_simpl\_2}
for which we provided the minimal unwinding depth necessary
to detect the bug, \Concbugassist 
timed out after 45 minutes, while our approach generated
a slice in 2.5 minutes for the non-simplified program.

Other static approaches for simplifying and summarizing concurrent
error traces include~\cite{DBLP:conf/popl/GuptaHRST15},
\cite{DBLP:conf/sas/HuangZ11}, \cite{DBLP:conf/sigsoft/JalbertS10},
and
\cite{DBLP:conf/cav/KashyapG08}. In~\cite{DBLP:conf/popl/GuptaHRST15},
an SMT solver and model enumeration is used to derive a symbolic
representation of 
\emph{all} reorderings of a given trace that violate a safety property,
which is then used to explain the bug.
Instead, we analyze a  single failing trace, ensuring that our encoding
explicitly captures which happens-before
relations are relevant for the faulty behavior.

Tools that attempt to minimize the number of context switches,
such as {\sc SimTrace} \cite{DBLP:conf/sas/HuangZ11} and {\sc
  Tinertia} \cite{DBLP:conf/sigsoft/JalbertS10},
are orthogonal to the approach presented in this paper. 

Many techniques for detecting race conditions
or atomicity/serializability violations
are geared towards specific bug
characteristics~\cite{FlanaganQ03,XuBH05,lu2008learning}.
Similarly, dynamic techniques such as {\sf Falcon}~\cite{ParkVH10}
and 
{\sf Unicorn}~\cite{Park12}
rely on bug patterns.
Our approach encodes data-dependencies rather than
relying on bug patterns or specific bug characteristics.
Recent work~\cite{Tab:abstraction} uses mining of
failing and passing traces to isolate erroneous sequences
of statements. Our technique only considers failing traces.

{\sc Afix}~\cite{DBLP:conf/pldi/JinSZLL11} and
{\sc ConcurrencySwapper}~\cite{DBLP:conf/cav/CernyHRRT13}
automatically fix concurrency-related errors.
The latter uses error invariants to generalize a linear error trace to a
partially ordered trace, which is then used to
synthesize a fix. This approach may potentially benefit from our more
fine-tuned trace encoding that enables error invariants to capture
concurrent data dependencies.
%
%

\section{Conclusion}
\label{sec:conclusion}

We proposed to augment error invariants with information
about inter-thread data dependency and hazards to
capture a broad range of concurrency bugs. Our technique
generates sound slices of concurrent error traces,
 enabling developers to quickly isolate and
focus on the relevant aspects of error traces.
We proved that the reported slices are sound and
sufficient to trigger the failure.
The experimental evaluation of our
prototype implementation showed that the approach is effective and
significantly reduces the amount of code that needs to be inspected.
%
%


\bibliographystyle{plain}
\bibliography{paper}{}

\clearpage
\appendix

\section{Proofs}
\label{sec:proofs}

\setcounter{theorem}{0}%
\begin{theorem}
  Let $P$ be a (concurrent) error trace for $\psi$ of length $k$ and let
  $I_0$, $I_1$ ,\ldots, $I_{k-1}$, $I_{k+1}$
  be error invariants (with $I_0=\true$ and
  $I_{k+1}=\false$) obtained from
  an inductive sequence interpolant
  for $\fsenc{n_1},\ldots,\fsenc{n_k},\psi$. Let $Q$ be the slice obtained
  from $P$ by removing each sub-path $P[i,j]$
  for which $I_{i-1}$ is inductive. Then $Q$
  is a sound control-sensitive slice for $P$.
\end{theorem}

\begin{proof}[(sketch)]
  Let $J_0,\ldots,J_{k+1}$
  be the sequence interpolant for $\fsenc{\tau}$ corres\-ponding
  to the error invariant
  $I_0$, $\ldots$, $I_{k+1}$.
  We show by induction over the length of $P$
  that it is also a sequence interpolant for
  $\enc(P)$ (where ${\sf x}_i\asgn\phi({\sf x}_j)$ is
  replaced by ${\sf x}_i\asgn{\sf x}_j$ in $P$).
  Our claim holds trivially for the base case,
  since $J_0=\true$ and $J_{k+1}=\false$.
  Assume that $J_{j-1}\wedge\enc(P)[j]\Rightarrow J_j$
  for $1\leq j\leq i$.  
  We distinguish the following cases:
  \begin{compactenum}
  \item If $\stmt_P[i]$ is an assignment without
    a $\phi$-function, then $\fsenc{n_i}=\enc(n_i)$
    and $J_i\wedge \enc(n_i)\Rightarrow J_{i+1}$.
  \item Let $\stmts_P[i]$ be the condition $R$. Since
    $\fsenc{n_i}=\true$, it holds that $J_i\Rightarrow J_{i+1}$,
    and therefore also $J_i\wedge\enc(n_i)\Rightarrow J_{i+1}$.
  \item\label{case:phi} Let $\stmts_P[i]$ be the assignment
    ${\sf x}_i\asgn\phi({\sf x}_j)$ at node $n_i$, and
    let $\stmts_P[j]$ be
    the statement by which ${\sf x}_j$ is assigned at node $n_j$.
    If $J_i\Rightarrow \guard(n_j)$ then
    $J_i\wedge \enc(n_i)\Rightarrow J_{i+1}$.
    Otherwise, $J_{i+1}$ cannot depend on 
    $\enc(n_i)$, and $J_i\Rightarrow J_{i+1}$;
    thus $J_i\wedge\enc(n_i)\Rightarrow J_{i+1}$.
  \end{compactenum}
  Since the error invariants
  $I_0,I_1,\ldots,I_{k},I_{k+1}$
  are derived from a sequence interpolant, and fragments
  $\\stmts_P[i,j]$ are only sliced if $I_{i-1}$ can replace $I_j$
  as an error invariant,
  we have for every
  $\stmts_Q[i]=\stmts_P[j]$ that the states reachable from
  an initial state via $\stmts_Q[1,i]$ are in $I_j$,
  i.e., $Q$ is a sound slice for the error trace $P$.

  It remains to show that $Q$ is control-sensitive. First,
  we observe the following properties of sequence interpolants:
  \begin{compactenum}
  \item[(a)] If an assignment $\stmts_P[i]={\sf x}_i\asgn e$ 
    is relevant, then $J_{i+1}$ must contain a non-redundant
    occurrence of ${\sf x}_i$, since
    otherwise $J_i\Rightarrow J_{i+1}$. 
  \item[(b)] Conversely, if $J_j$ contains a non-redundant occurrence
    of ${\sf x}_i$ and $J_{j+1}$ does not, then $\stmts_P[j]$ must be 
    relevant.
  \item[(c)] If $\stmts_P[j]$ 
    is an assignment defining ${\tt x}_j$ and $I_j$ refers to
    a previous version ${\tt x}_i$ defined before $\stmts_P[j]$,
    then $I_{j+1}$ must not contain ${\tt x}_i$, since $I_{j+1}$
    otherwise violates the condition that interpolants must only refer
    to shared variables; in particular, that means that $\stmts_P[j]$
    must not be sliced.
  \end{compactenum}
  
  Assume that $\stmts_P[i]={\sf x}_i\asgn e$ is relevant and in scope of
  an assumption $\stmts_P[j]$. Then $J_{i+1}$ contains ${\sf x}_i$ by
  observation (a) above. Since the program execution $P$ 
  reaches the exit node of the main thread before $\psi$ is asserted,
  $P$ must eventually traverse an arbiter node (representing
  the end of the scope of $\stmts_P[j]$ or a context switch)
  annotated with the statement $\stmts_P[l]$. 
  The value of ${\sf x}_i$ is propagated either directly (in
  which case $J_l$ contains ${\sf x}_i$) or
  via a sequence of assignments (in which case $J_l$ contains
  a variable ${\sf y}$ which replaced ${\sf x}_i$ by means
  of a relevant assignment in the same scope, as explained
  in observation (b) above) to $\stmts_P[l]$. Consequently,
  $J_l$ contains ${\sf x}_i$ (or {\sf y}).

  The statement $\stmts_P[l]$ contains the premise $\guard(n_j)$
  and replaces ${\sf x}_i$ ({\sf y}, respectively) with a 
  newer version of the variable. Accordingly, $J_{l+1}$ must
  not contain ${\sf x}_i$ (or {\sf y}). Therefore 
  $J_l\Rightarrow\guard(n_j)$.

  Since $\guard(n_i)$ is not modified by $\stmts_P[i,l]$, it must
  also be implied by $J_{i-1}$ in order to be propagated to $J_{l-1}$.
  By the soundness argument above, $I_{i-1}$ is established by a
  prefix of $\stmts_Q[1,j]$.
\end{proof}

\begin{theorem}
  Let $P$ be a concurrent error trace and let
  $Q$ be the slice obtained
  from $P$ as explained in Section \ref{subsec:data_dep}.
  Then $Q$ is a sound hazard-sensitive slice of $P$.
\end{theorem}

\begin{proof}[(sketch)]
  Assume that $\stmts_Q[k]=\stmts_P[i]$ (at node $n_i$) and there is
  an inter-thread data-dependency between $\stmts_P[j]$ (at node $n_j$)
  and $\stmts_P[i]$.
  \begin{itemize}
  \item Assume that $\stmts_P[i]$ is a read access to {\sf x}, i.e.,
    $\rdvar{\sf x}{n_i}$. RAW dependencies are readily handled
    by the SSA encoding. The remaining WAR dependencies
    are encoded in the $\pi$-function of the
    arbiter node $n_l$ for $n_i$
    (which assigns the variable ${\sf x}_i$ used in
    $\stmts_P[i]$). If $\In_i$ refers to ${\sf x}_i$ (i.e., the value of
    ${\sf x}_i$ is relevant in $\stmts_P[i]$) then the $\pi$-node is included
    in the slice (if not, the data-dependency has no impact on the
    failure of the trace).
    
    Formula (\ref{eq:pi-encoding2}) requires that
    every node $m\in P$ that writes to {\sf x} is either visited
    before the most recent write access to {\sf x} or after the
    read access $\stmts_P[i]$. Assume that $\hb{m}{n_j}$ in $P$.
    Then $\war{\sf x}{n_i}{m}$ evaluates to \false, and the
    interpolant $\In_l$ must imply $\wrvar{\sf x}{m}$, since
    otherwise $\war{\sf x}{m}{n_j}$ in the premise
    (\ref{eq:pi-encoding2}) of Formula (\ref{eq:pi-encoding1}) 
    cannot be discharged. The
    predicate $\wrvar{\sf x}{m}$ can only be introduced into
    the interpolation sequence through $\hsenc{m}$, and therefore
    node $m$ cannot be sliced away. If $\hb{n_i}{m}$ in $P$,
    then the premise of $\Out_l$ can only be discharged
    by $\wrvar{\sf x}{m}$ contributed by $\hb{n_i}{m}$.
    Consequently, if node $m$ is not included, the final
    interpolant cannot be $\false$.
  \item Assume that $\stmts_P[i]$ is a write access to {\sf x}.
    Then there must also be a relevant read access to {\sf x} in
    $Q$. The encoding of the corresponding $\pi$-node
    will enforce that all write accesses conflicting with $\stmts_P[i]$
    are included in the trace.
  \end{itemize}
\end{proof}


\section{Case Study: Lock-free Concurrent Data Structure}
\label{sec:case_study_pool}

\newcommand{\poolins}{\textsf{pool\_ins}\xspace}
\newcommand{\poolrem}{\textsf{pool\_rem}\xspace}
\newcommand{\EMPTY}{\textsf{EMPTY}\xspace}

\begin{figure}
\centering
\begin{minipage}{\linewidth}
\begin{minipage}[t]{.65\linewidth}
\begin{lstlisting}[language=C,captionpos=b,%frame=single,
%label=lst:pool,
%caption={Faulty pool based on Treiber stacks},float=tb,
basicstyle=\sffamily\scriptsize,
commentstyle=\itshape\color{darkgray},
keywordstyle=\bfseries,
%multicols=2,
numbersep=5pt,
morekeywords={assert},
%numbers=left,numberstyle=\sffamily\tiny\color{gray},numberblanklines=false,
tabsize=2]
treiber_stack_t* ts[2];
void pool_ins(int v) { 
	// we assume that v != EMPTY
	int idx = random()%2;
	ts_push(&ts[idx], v);
}


int pool_rem() {
	int idx = random()%2;
  for (int i = 0; i < 2; i++) {
  	int v = ts_pop(&ts[(idx+i)%2]);
  	if (v != EMPTY) return v;
  }
	return EMPTY;
}
\end{lstlisting}
\end{minipage}%
\begin{minipage}[t]{.35\linewidth}
\begin{lstlisting}[language=C,captionpos=b,%frame=single,
%label=lst:pool,
%caption={Faulty pool based on Treiber stacks},float=tb,
basicstyle=\sffamily\scriptsize,
commentstyle=\itshape\color{darkgray},
keywordstyle=\bfseries,
%multicols=2,
numbersep=5pt,
morekeywords={assert},
%numbers=left,numberstyle=\sffamily\tiny\color{gray},numberblanklines=false,
tabsize=2]
void thread1() {
	pool_ins(1);
	pool_rem();
}

void thread2() {
	pool_ins(2);
	int v = pool_rem();
	assert(v != EMPTY);
}
\end{lstlisting}
\end{minipage}
\end{minipage}
\vspace{-1em}
\caption{Faulty thread pool implementation based on Treiber stacks\label{lst:pool}}
\end{figure}

In the following, we discuss benchmark \textsf{pool\_simple\_2} from
Table~\ref{tbl:benchmarks} in more depth as it demonstrates that, in
general, both control and hazard-sensitive information is needed to
obtain useful bug explanations.

Benchmark \textsf{pool\_simple\_2} was provided by Andreas Haas at
University of Salzburg, as a real-world example of a linearizability
bug in concurrent data structures. It comprises a faulty
implementation of a concurrent data structure that stores objects in a
pool. Listing~\ref{lst:pool} shows a simplified version of the actual
source code that we analyzed. In order to reduce contention, objects
that are inserted into this pool are stored in two different
stacks~\textsf{ts[0]} and~\textsf{ts[1]}.  Each time \poolins is called, a stack
will be picked randomly and the passed value will be stored in the
selected stack.  Thereby, the amount of conflicting operations from
different threads at each concurrent data structure is reduced.  In
order to further reduce contention, one can add more stacks.

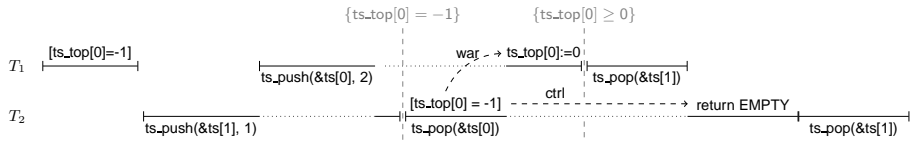
\begin{figure}[tb]
  \vspace*{-1em}
	\centering
	\scalebox{0.67}{
	\begin{tikzpicture}
          \node (T1) at (-2, 0) {$T_1$};
          \node (T2) at (-2, -1) {$T_2$};

                \draw[|-|] (-1.5,0)--node[above]{\sf [ts\_top[0]=-1]}(0.4,0);

                \draw[|-] (0.5,-1) -- node[below] {\sf
                  ts\_push(\&ts[1], 1)} (2.8,-1);
                \draw[dotted] (2.8,-1)--(5.1,-1);
                \draw[-|] (5.1,-1)--(5.6,-1);

                \draw[dashed,gray] (5.65,.75)node[above]{$\{\mathsf{ts\_top[0]=-1}\}$}--(5.65,-1.5);

                \draw[|-] (5.7,-1) -- node[below] {\sf
                  ts\_pop(\&ts[0])} node[above] {\sf [ts\_top[0] = -1]}
                (7.7,-1); 

                \draw[|-] (2.8,0) -- node[below] {\sf
                  ts\_push(\&ts[0], 2)} (5.1,0);
                \draw[dotted] (5.3,0)--(7.7,0);
                \draw[-|] (7.7,0)-- node[above] {\sf
                  ts\_top[0]:=0} (9.2,0);

                \draw[dashed,gray]
                (9.25,.75)node[above]{$\{\mathsf{ts\_top[0]\geq
                    0}\}$}--(9.25,-.7);
                \draw[dashed,gray] (9.25,-1.1)--(9.25,-1.5);

                \draw[|-|] (9.3,0) -- node[below] {\sf
                  ts\_pop(\&ts[1])} (11.3,0);

                \draw[dotted] (7.7,-1) -- (11.3,-1);
                \draw[-|] (11.3,-1)-- node[above] {\sf
                  return EMPTY} (13.5,-1);
                \draw[|-|] (13.5,-1) -- node[below] {\sf
                  ts\_pop(\&ts[1])} (15.7,-1);

                \draw[dashed,->] (7.8,-.75)--node[near
									start,above]{\sf{ctrl}} (11.3,-.75);
                \path[dashed,->] (6.5,-.5) edge [bend left=30] node[above]{\sf{war}} (7.6,.25);
		
        \end{tikzpicture}}
	\caption{Error trace of program in Fig.~\ref{lst:pool} with dependencies ([\dots] denote conditions)
	\label{fig:case_study_pool}}
\end{figure}

The \poolrem operation of the pool may incorrectly return the
designated value \EMPTY although the pool is not empty (checked via
the assertion in \textsf{thread2}).  The problem can occur when \poolrem is
called and, for example, stack~\textsf{ts[1]} is empty but \textsf{ts[0]} is not.
Figure~\ref{fig:case_study_pool} shows a corresponding faulty
program execution.  
We describe the explanation our tool provides for one of the faulty
traces generated for the pool example.  To highlight the problematic
dependencies in the execution, we need to inspect the trace at
instruction level, as the interferences are not reflected at the level
of the overlapping procedure calls.  The implementation of the {\sf
  treiber\_stack} data-structure uses the entry {\sf ts\_top[i]} to
store the index of the top element of the $\mathsf{i}^{\text{th}}$
stack. The value of {\sf ts\_top[i]} is $-1$ if the corresponding
stack is empty. The write access to the actual stack is implemented
using an atomic \emph{compare-and-swap} operation (guaranteeing
exclusive access to the top of the stack), which only succeeds if no
other thread interferes with the write operation. As shown in
Listing~\ref{lst:pool}, {\sf pool\_rem} iterates over all stacks to
check whether one of them contains an element that can be removed.

In the generated trace, the assertion that {\sf ts\_top[i]} must be
$-1$ for all stacks if the pool is reported to be empty fails.  The
statements in Figure~\ref{fig:case_study_pool} are part of the slice
reported by our tool and highlight the underlying problem: thread
$T_1$ pushes an element onto stack $0$ ({\sf ts\_top[0]:=0})
\emph{after} thread $T_2$ has determined that the stack is empty.
This is captured by the anti-dependency between the statements {\sf
  [ts\_top[0]=-1]} and {\sf ts\_top[0]:=0} (denoted by the \textsf{war} edge).  
Thread $T_1$ then proceeds to remove the element previously
pushed by $T_2$ onto stack $1$. Consequently, thread $T_2$ finds stack
1 empty and reports that the pool is empty (based on a stale value of
{\sf ts\_top[0]}), even though stack 0 still contains one
element. This is captured by the control-dependency between {\sf
  [ts\_top[0]=-1]} and {\sf return EMPTY} (denoted by the \textsf{ctrl} edge). 
Thus, even though the assignment {\sf ts\_top[0]:=0} is
implemented as an atomic compare-and-swap operation in the actual
code, this does not guarantee correctness of the lock-free
implementation: the operation {\sf pool\_rem} is not
\emph{linearizable}, since its effect is not instantaneous.  

The core of the problem is accurately reflected by the control-sensitive
slice generated by our tool:
{\sf return EMPTY} is necessary to satisfy the premise
of the assertion, and {\sf ts\_top[0]:=0} must be included to
contradict the conclusion. The return statement is
control-dependent on {\sf [ts\_top[0] = -1]}, and the
explanation therefore includes the initialization of 
{\sf ts\_top[0]}. 

While the control-sensitive slice that our tool computes
does \emph{not} explicitly include the condition
{\sf [ts\_top[0] = -1]}, it is 
reflected by the error invariant $\mathsf{ts\_top[0]}=-1$. 
This information is explicit in the hazard-sensitive slice generated
by our tool, which includes the anti-dependent statements  
{\sf [ts\_top[0] = -1]} in thread $T_2$ and {\sf ts\_top[0]:=0} in
thread $T_1$. Notably, the control and hazard-sensitive slice is only
marginally longer than the control-sensitive slice: the former
contains 264 instructions, whereas the latter contains 255 instructions,
or 28\% of the 924 instructions of the original trace. In addition,
our tool drops roughly 44\% of the variables of the original trace.


\end{document}